\documentclass[a4paper,11pt,english,reqno]{amsart}
\usepackage{amsmath,amssymb,amsfonts,dsfont,amsthm,upgreek,bm,mathtools}    
\usepackage[utf8]{inputenc}
\usepackage[english]{babel}
\usepackage[T1]{fontenc} 
\usepackage{graphicx,xcolor,pgfkeys}
\usepackage{float}
\usepackage[font=small,labelfont=bf]{caption}
\usepackage[a4paper, margin=2.4cm]{geometry}
\usepackage{tikz,pgfplots}
\usepackage{tkz-fct}
\usepgfplotslibrary{polar}
\usetikzlibrary{patterns, arrows}

\usepackage{enumerate}
\usepackage{enumitem}
\usepackage[normalem]{ulem}
\newcommand\redsout{\bgroup\markoverwith{\textcolor{red}{\rule[0.5ex]{2pt}{1pt}}}\ULon}
\newcommand{\stkout}[1]{\ifmmode\text{\redsout{\ensuremath{#1}}}\else\redsout{#1}\fi}
\usepackage{scalerel,stackengine}
\stackMath
\newcommand\reallywidecheck[1]{%
\savestack{\tmpbox}{\stretchto{%
  \scaleto{%
    \scalerel*[\widthof{\ensuremath{#1}}]{\kern-.6pt\bigwedge\kern-.6pt}%
    {\rule[-\textheight/2]{1ex}{\textheight}}
  }{\textheight}%
}{0.6ex}}%
\stackon[1pt]{#1}{\scalebox{-0.8}{\tmpbox}}%
}

\newcommand{\prob}{\mathds{P}}

\newcommand{\e}{\mathrm{e}}

\newcommand{\tr}{\mathrm{tr\,}}

\newcommand{\supp}{\mathrm{supp\,}}

\newcommand{\Ima}{\mathrm{Im\,}}
\newcommand{\Rea}{\mathrm{Re\,}}
\newcommand{\dist}{\mathrm{dist\,}}
\newcommand{\diag}{\mathrm{diag}}
\newcommand{\Op}{\mathrm{Op}}
\newcommand{\mO}{\mathcal{O}}
\newcommand{\C}{\mathds{C}}
\newcommand{\R}{\mathds{R}}

\newcommand{\N}{\mathds{N}}
\newcommand{\Z}{\mathds{Z}}

\newcommand{\wt}{\widetilde}

\newcommand{\wh}{\widehat}

\newcommand{\cV}{\mathcal{V}}

\renewcommand{\ge}{\geqslant}
\renewcommand{\geq}{\geqslant}
\renewcommand{\le}{\leqslant}
\renewcommand{\leq}{\leqslant}

\newtheorem{thm}{Theorem}
\newtheorem{corollary}[thm]{Corollary}

\newtheorem{prop}[thm]{Proposition}
\newtheorem{lem}[thm]{Lemma}

\newtheorem{definition}[thm]{Definition}
\newtheorem{rem}[thm]{Remark}
\numberwithin{equation}{section}
\numberwithin{thm}{section}

\usepackage{hyperref}
\hypersetup{pdfborder=0 0 0, 
	    colorlinks=true,
	    citecolor=blue,
	    linkcolor=blue,
	    urlcolor=blue,
	    pdfauthor={Simon Becker, Izak Oltman, Martin Vogel}
	   }
\title[Absence of small magic angles]{Absence of small magic angles for disordered tunneling potentials in twisted bilayer graphene}

\author{Simon Becker}
\address[Simon Becker]{ETH Zurich, 
Institute for Mathematical Research, 
R\"amistrasse 101, 8092 Zurich,
Switzerland}
\email{simon.becker@math.ethz.ch}
\author{Izak Oltman}
\address[Izak Oltman]{Department of Mathematics, 
University of California, 
Berkeley, CA 94720, USA.}
\email{ioltman@berkeley.edu}
\author{Martin Vogel}
\address[Martin Vogel]{Institut de Recherche Math{\'e}matique Avanc{\'e}e - UMR 7501, 
Universit{\'e} de Strasbourg et CNRS, 7 rue René-Descartes, 67084 Strasbourg Cedex, France.}
\email{vogel@math.unistra.fr}
\date{\today}
\usepackage[most]{tcolorbox}

\begin{document}
\begin{abstract}
We consider small random perturbations of the standard high-symmetry 
tunneling potentials in the Bistritzer-MacDonald Hamiltonian 
describing twisted bilayer graphene. Using methods developed by 
Sj\"ostrand for studying the spectral asymptotics of non-selfadjoint 
pseudo-differential operators, we prove that for sufficiently 
small twisting angles the Hamiltonian will not exhibit a flat band 
with overwhelming probability, and hence the absence of the so-called 
\textit{magic angels}. Moreover, we prove a probabilistic Weyl law for 
the eigenvalues of the non-selfadjoint tunneling operator, subject to 
small random perturbations, of the Bistritzer-MacDonald Hamiltonian in the 
chiral limit.
\end{abstract}
\maketitle
\tableofcontents
\section{Introduction}
Twisted bilayer graphene is a stacked and twisted two-dimensional carbon 
material that exhibits a variety of strongly correlated electron phenomena 
such as superconducting phases \cite{Cao}. In the one-particle band structure, 
the existence of strongly correlated phases is indicated by the occurrence of 
flat bands. 
The purpose of this article is to study the stability of flat bands under 
small random perturbations of the tunneling potentials. Such perturbations 
adequately reflect material impurities e.g. due to internal strain effects 
(lattice relaxations) \cite{NamKo17}.

The stability of flat bands under random perturbations depends sensitively 
on the nature of the disorder. If the disorder is signed, then classical 
Wegner-type estimates rule out the presence of flat bands under disorder, 
since the integrated density of states does not exhibit any jump 
discontinuities. This has been implemented for magnetic Schr\"odinger 
operators \cite{CHK,GKM,GKS,GKS2} and for twisted bilayer graphene by the 
authors in \cite{BOV23}.

In this article, we deal with a more realistic scenario, where we study 
random perturbations of the standard high-symmetry tunneling potentials 
in the Bistritzer-MacDonald Hamiltonian \cite{BM11} for twisted bilayer 
graphene, see \cite{CGG,Wa22}. Since the potential perturbations are not 
signed, it is not immediate that they can sufficiently perturb the spectrum 
to destroy the flat band. Our main result, Theorem \ref{thm1}, shows 
for random tunneling potential perturbations, there are no flat bands with 
overwhelming probability.

Let us start by introducing the Hamiltonian. Thus, let $h \in ]0,1]$ be 
proportional to the physical twisting angle, then the Bistritzer-MacDonald 
Hamiltonian is a semiclassical matrix-valued first order differential 
operator of the form
\begin{equation}\label{eq:Hamiltonian}
H_{\text{BM}}(w,h) 
= \begin{pmatrix} w C & D_h \\ D_h^* & w C \end{pmatrix}
\end{equation}
acting on $L^2(\mathbb C;\mathbb C^4)$ with domain $H^1(\mathbb C;\mathbb C^4).$

Letting $D_{x} =\frac{1}{2}( D_{x_1} -i D_{x_2})$ and 
$D_{x_j} = -i \partial_{x_j}$, the matrix-valued entries in 
\eqref{eq:Hamiltonian} are
\begin{equation}\label{eq1.0}
	D_h := 
	\begin{pmatrix}
		2hD_{\overline{x}} & U(x) \\
		U(-x) & 2hD_{\overline{x}} \\
	\end{pmatrix} \text{ and } C
	=\begin{pmatrix} 0 & V(x)\\ 
		\overline{V(x)} & 0
	\end{pmatrix},
\end{equation}
where he tunnelling potentials $U,V$ are smooth functions which are 
characterized, for $a_j = \frac{4}{3}\pi i \omega^j$ 
and $\omega = \operatorname{exp}(2\pi i /3)$, by
\begin{equation}
  \label{eq:symmU}
    \begin{split}
      &V ( x + \mathbf a_j ) = \bar \omega V ( x ) ,\ \ V ( \omega x ) 
	  	= V ( x ) , \ \  \overline{ V ( x ) } = V ( - x ) , \ \ { V ( \bar x ) } 
	  	= V (-x ) , \\
      &U ( x + \mathbf a_j ) = \bar \omega U ( x ) ,  \ \ U ( \omega x ) 
	  = \omega U ( x ) , \ \  \overline{ U ( \bar x ) } 
	  = U ( x ).
    \end{split}
\end{equation}
In this article, we focus on the chiral limit \cite{TKV19}. This limit is 
obtained by setting $w_0 \equiv 0$ in the Hamiltonian \eqref{eq:Hamiltonian} 
\begin{equation*}
	H_{\text{chiral}} 
	= \begin{pmatrix} 0 & D_h \\ 
		D_h^* & 0 
	\end{pmatrix}.
\end{equation*}
Since $H_{\text{chiral}}$ is periodic with respect to the lattice 
\begin{equation}\label{eq:Gamma}
	\Gamma:=4\pi(i\omega \Z\oplus i\omega^2\Z),
\end{equation}
we can apply the Bloch-Floquet transform, see \cite[Sec.2.3]{BEWZ22}, to 
obtain an equivalent family of operators parametrized by $k \in \C$ 
on $L^2(\C/\Gamma;\C^4)$ with domain $H^1_h(\C/\Gamma;\C^4)$
\begin{equation}\label{eq:Hk}
H_{\text{chiral}}(k) 
= \begin{pmatrix} 0 & D_h+ hk \\ 
	D_h^* +h\bar k& 0 
\end{pmatrix}. 
\end{equation}
We see $\C/\Gamma$ as a smooth compact manifold equipped with a smooth 
positive density of integration $dx$. A natural choice would be the 
Riemannian volume density inherited from 
$\C$.

The operator \eqref{eq:Hk} satisfies the chiral symmetry 
\begin{equation*}
	\operatorname{diag}(1,-1) H_{\text{chiral}}(k) \operatorname{diag}(1,-1) 
	= - H_{\text{chiral}}(k).
\end{equation*}
This implies that eigenvalues of $H_{\text{chiral}}(k)$ come in pairs 
$E_{-n}(k)=-E_n(k)$ with 
\begin{equation*}
	\dots \le E_{-2}(k) \le E_{-1}(k)\le 0 \le E_1(k)\le E_2(k)\le \dots .
\end{equation*}
When $E_1(k)\equiv 0$ we say that $H_{\text{chiral}}$ exhibits a flat band 
at energy zero. Since $D_h$ is a Fredholm operator of index $0$, see e.g. 
\cite[Proposition 2.3]{BEWZ22}, one concludes that 
\begin{equation}\label{eq:equiv_cond} 
	E_1(k) \equiv 0 \Longleftrightarrow \operatorname{Spec}(D_h)
	=\C.
\end{equation}
A twisting angle proportional to $h$ at which \eqref{eq:equiv_cond} holds 
is referred to as a \textit{magic angle} \cite{BEWZ22}.  
In the present article, we shall use \eqref{eq:equiv_cond} to study the 
equivalent magic angle condition $\operatorname{Spec}(D_h)=\C$ under 
random perturbations by off-diagonal potentials of $D_h$.
In Theorem \ref{thm1}, we show that the lowest singular value of 
a suitable random perturbation $D_h-z$ being at least an exponentially 
small distance away from zero is overwhelmingly high. Thus, $z$ is not 
in the spectrum of the perturbation of $D_h$ with great probability 
and thus there cannot be any flat band, or equivalently, a magic angle. 
\par
One might object that an exponentially flat band would in practise still 
look fairly flat. But in the case of the chiral model, it is known 
\cite{BEWZ22} that there are $\ge \mathcal O(1/h)$ many bands that 
are $\mathcal O(e^{-c/h})$ close to zero energy. Thus, the bound in 
Theorem \ref{thm1} implies that the designated flat band may 
mix with the other exponentially small bands which do not carry a 
non-zero Chern number. 
\\
\par
In general, there may be topological obstructions to destruct a flat 
band of a Dirac operator by potential-type perturbations. This is why 
the study of non-signed perturbations is subtle.
An example where this happens is the magnetic Dirac operator 
\begin{equation*}
	H_{\text{magnetic}} 
	= \begin{pmatrix} 
		0 & 2D_{\bar x} - A(x) \\ 
		2D_{x}-\overline{A(x)} & 0 
	\end{pmatrix}
\end{equation*}
with magnetic potential $A(x)=A_1(x)+i A_2(x).$ If we choose a magnetic 
potential $A$ that gives rise to a constant magnetic field 
$B = \partial_1 A_2 - \partial_2 A_1 >0,$ then we can apply the 
magnetic Bloch transform to study its spectrum on a compact domain 
$\C/\Gamma$ on which the magnetic flux is commensurable, i.e. 
$\phi:=B \vert \C/\Gamma \vert \in \Z_+$ with so-called magnetic 
boundary conditions. Details of this construction can for instance be found 
in \cite{BZ24}.

Similar to \eqref{eq:equiv_cond}, one has that 
$\operatorname{Spec}_{L^2(\C/\Gamma)}(2D_{\bar x} - A) = \C.$ However, 
the Fredholm index of $2D_{\bar x}-A(x)$ on $\C/\Gamma$ is equal to 
$\phi \neq 0$, see \cite{BZ24}. Since this operator $2D_{\bar x}-A(x)$ 
is elliptic on a compact domain, its Fredholm index is invariant under 
bounded potential perturbations. This implies that its spectrum is equal 
to the entire complex plane even under arbitrary bounded potential 
perturbations. In this sense, the flat bands of a magnetic Dirac operator, 
the so-called Landau levels, cannot be destroyed by random perturbations 
of the magnetic potential. 
\section{Main results}
Let $h\in]0,1]$ and consider the unbounded semiclassical differential 
operator 
\begin{equation}\label{eq1}
	D_h := 
	\begin{pmatrix}
		2hD_{\overline{x}} & U(x) \\
		U(-x) & 2hD_{\overline{x}} \\
	\end{pmatrix}: 
	L^2(\C/\Gamma;\C^2) \to L^2(\C/\Gamma;\C^2), 
\end{equation}
with $U \in C^\infty(\C/\Gamma;\C)$ as in \eqref{eq:symmU} and 
equipped with the domain $H_h^1(\C/\Gamma;\C^2)$, making it a closed densely 
defined unbounded operator. It was observed in \cite[Proposition 2.3]{BEWZ22} 
that $D_h$ is Fredholm operator of index $0$. 
\subsection{Random tunneling perturbation}
Let $\widetilde{P}_1,\widetilde{P}_2$ be $h$-independent elliptic positive 
second order differential operators on $\C/\Gamma$ with smooth coefficients. 
Put $P_j = h^2 \widetilde{P}_j$, $j=1,2$, and let $\{\psi_n^{1}\}_{n\in\N},
\{\psi_n^{2}\}_{n\in\N}\subset L^2(\C/\Gamma;\C)$ be two orthonormal bases 
composed of eigenfunctions of $P_1$ and $P_2$, respectively, so that 
\begin{equation}\label{eq2}
	P_j \psi_n^j = \mu_{n,j}^2\psi_n^j, \quad 
	\mu_{n,j}\geq 0. 
\end{equation}
For $L\gg1$ and $D_j=D_j(L)>0$ we consider the potentials 
\begin{equation}\label{eq3}
	q^j_\alpha(x) = \sum_{0\leq \mu_{n}^j\leq L} \alpha_n \psi_{n}^j(x), 
	\quad \alpha \in \C^{D_j}.
\end{equation}
The Weyl law, cf. \eqref{eq:sa8.00}, for the eigenvalues of elliptic 
self-adjoint second order differential operators yields 
\begin{equation}\label{eq3.1}
	D_j \asymp L^2 h^{-2}.
\end{equation}
We then know from standard Sobolev estimates, see Proposition \ref{app:prop4} 
and \eqref{eq:sa8.0}, that for $j=1,2$ and $s>1$  
\begin{equation}\label{eq4.1}
	\| q^j_\alpha\|_{H_h^{s}(\C/\Gamma)} \leq \mO_s(1)L^{s}\|\alpha\|_{\C^{D_j}},
	\quad \| q^j_\alpha\|_{L^\infty(\C/\Gamma)} \leq \mO_s(1) h^{-1}L^{s}\|\alpha\|_{\C^{D_j}}. 
\end{equation}
Here, $H_h^{s}(\C/\Gamma)$ is a semiclassical Sobolev space. See Appendix 
\ref{App:MatrixPseudo} for a brief review.  
Constraints on $L$ and $\|\alpha\|_{\C^{D_j}}$ will be specified later on. 
The potentials \eqref{eq3} give rise to a tunneling potential 
\begin{equation}\label{eq5}
	Q_\gamma = 
	\begin{pmatrix}
		0 & q^1_\alpha(x) \\
		-q^2_\beta(x) & 0 \\
	\end{pmatrix}
	:L^2(\C/\Gamma;\C^2)\to L^2(\C/\Gamma;\C^2),  
\end{equation}
with $\gamma=(\alpha,\beta) \in \C^{D_1}\times \C^{D_2} \simeq \C^{D}$, 
$D=D_1+D_2$. By \eqref{eq4.1} and standard Sobolev estimates, cf. 
Proposition \ref{app:prop4},
\begin{equation*}
	\|Q_\gamma\|_{L^2(\C/\Gamma;\C^2)\to L^2(\C/\Gamma;\C^2)}
	\leq 
	\|Q_\gamma\|_{L^\infty(\C/\Gamma;\C^{2\times 2})} 
	= 
	\mO_s(1) h^{-1}L^{s}\|\gamma\|_{\C^{D}}.
\end{equation*}
Let $\gamma$ be real or complex random vector in $\R^{D}$ or 
$\C^{D}$, respectively, with joint probability law
\begin{equation}\label{eq:probaM}
	\gamma_*(d \prob) =  
	Z_h^{-1}\, \mathbf{1}_{B(0,R) }(\gamma) \,\e^{\phi(\gamma)} L(d\gamma),
\end{equation} 
where $Z_h>0$ is a normalization constant, $B(0,R)$ is either the real 
ball $\Subset \R^{D}$ or the complex ball $\Subset \C^{D}$ of radius 
$R \gg 1$, and centered at $0$, $L(d\gamma)$ denotes the Lebesgue
measure on either $\R^{D}$ or $\C^{D}$ and $\phi \in C^1$ with 
\begin{equation}\label{eq:probaM2}
	\| \nabla_\gamma \phi \| = \mO(h^{- \kappa_4}) 
\end{equation}
uniformly, for an arbitrary but fixed $\kappa_4\geq 0$. 
\\
\par 
Fix $s>1$ and $\varepsilon \in ]0,s-1[$, and for 
$C>0$ large enough, we fix 
\begin{equation}
	L = C h^{\frac{5}{s-1-\varepsilon}}, \quad 
	Ch^{-2-\frac{5s}{s-1-\varepsilon}} \leq R \leq Ch^{-\kappa_3}, 
	\quad 
	\kappa_3\geq 2+ \frac{5(1+\varepsilon)}{s-1-\varepsilon}, 
\end{equation}
and
\begin{equation}
	\kappa_1= 1+ \frac{5s}{s-1-\varepsilon} + \kappa_3, 
	\quad
	\kappa_5 = \kappa_3+\kappa_4 +2 + \frac{10}{s-1-\varepsilon}.
\end{equation}
\begin{rem}
	One example for a random vector $\gamma$ with law \eqref{eq:probaM} 
	is a truncated complex or real Gaussian random variables with expectation $0$,  
	and uniformly bounded covariances. The covariance matrix $\Sigma \in \C^{D\times D}$ 
	then satisfies $\Sigma>0$ and $\|\Sigma\|=\mO(\sqrt{D})$. In this case 
	$\phi(\gamma) = -\langle \Sigma^{-1} \gamma |\gamma\rangle$ 
	and \eqref{eq:probaM2} holds with 
	$\kappa_4 = 1+ \frac{5}{s-1-\varepsilon}+\kappa_3$.
\end{rem}
\subsection{Absence of small magic angels with overwhelming probability}
Let $\tau_0\in]0,\sqrt{h}]$, let $C_0>0$ be large enough, and 
consider the perturbed operator 
\begin{equation}\label{eq6}
	D^\delta_h:= D_h + \delta h^{\kappa_1} Q_\gamma, 
	\quad \delta = \frac{\tau_0 h^{\kappa_1+2}}{C_0},
\end{equation}
which, equipped with the domain $H^1_h(\C/\Gamma;\C^2)$, is 
a closed densely defined operator $L^2 (\C/\Gamma;\C^2)\to 
L^2 (\C/\Gamma;\C^2)$. 
\par 
Fix $z\in \C$. Then  
\begin{equation*}
	S := (D_h-z)^*(D_h-z)
\end{equation*}
is selfadjoint on $H_h^2(\C/\Gamma;\C^2)$. Since $S$ is a 
positive selfadjoint operator with domain that injects compactly 
into $L^2(\C/\Gamma;\C^2)$, 
it follows that $S$ has compact resolvent, and therefore its spectrum 
contains only isolated eigenvalues of finite multiplicity. Let 
$N:=N(\tau_0^2)$ be the number of eigenvalues $t_j^2$ of $S$ in $[0,\tau_0^2]$, 
i.e. 
\begin{equation*}
	0\leq t_1^2 \leq \dots \leq t_N^2 \leq \tau_0^2 
	< t_{N+1}^2 \leq \dots,  
\end{equation*}
with $t_j(D_h-z)=t_j \geq 0$. In \eqref{fa:eq13} we prove an upper bound 
on $N$ which in combination with the corresponding lower bound on $N$ proven 
in \cite[Theorem 5]{BEWZ22} gives 
\begin{equation*}
	\frac{1}{C} h^{-1} \leq N(\tau_0^2) \leq C h^{-1}, 
\end{equation*}
for some $C>0$.
\begin{thm}\label{thm1}
	Let $0 < \varepsilon <  \exp( -C_1 h^{-2} \varepsilon_0(h))$, $C_1>1$ 
	sufficiently large, let $N=N(\tau_0^2)$, and let 
	\begin{equation*}
		\varepsilon_0(h) :=  C(\log \tau_0^{-1} + (\log h^{-1})^2)(h+h^2\log h^{-1}).
	\end{equation*}
	with $C>0$ large enough. Then, for $h>0$ small enough
	\begin{equation*}
		\prob \left( 
			 t_1(D^\delta_h-z) \geq \frac{\varepsilon}{8(C\tau_0)^{N-1}}
			\right)
		\geq 1- C h^{-\kappa_5}\varepsilon_0(h)
		\exp\left( - \frac{h^2}{C\varepsilon_0(h)} \log \frac{1}{\varepsilon }
		\right).
	\end{equation*} 
\end{thm}
Let us give some comments on this result: The spectra of non-selfadjoint 
operators -- as opposed to the spectra of their selfadjoint counterparts -- 
can be highly sensitive to even small perturbations. Early works studying 
this phenomenon can be traced back to works of Trefethen 
\cite{Tr97,TrEm05} in numerical analysis. In the context of semiclassical 
non-selfadjoint pseudo-differential operators spectral instability 
can be very pronounced. We refer the reader to the works by Davies 
\cite{Da99b,Da99}, Zworski \cite{Zw01}, Dencker-Zworski-Sj\"ostrand 
\cite{NSjZw04} and Pravda Starov \cite{Pr06,Pr08} for a characterization 
of the zone of spectral instability in the complex plane. 
\par
In \cite{Ha06,Sj09,Sj10a} Hager and Sj\"ostrand proved a result similar to 
Theorem \ref{thm1} for semiclassical elliptic differential operators $P(x,hD_x)$ 
with principal symbol $p(x,\xi)$ subject to small perturbations of the 
form \eqref{eq3}. However, in their work one requires fundamentally that the 
principal symbol satisfies 
the symmetry $p(x,\xi)=p(x,-\xi)$. In the present work we can circumvent 
this assumption due to the $2\times2$ matrix structure of $D_h$ which 
allows for the use of random tunneling potentials of the form \eqref{eq5}. 
\par
Another stark difference to these previous works is that the operator $D_h$ 
is somewhat pathological from a spectral theory point of view. Indeed, for 
certain discrete values of $h$ the spectrum of operator $D_h$ can be the 
entire complex plane $\C$ (precisely when $h$ is proportional 
to a magic angle). This possibility has been excluded in previous 
works by \cite{Ha06,Sj09,Sj10a}. However, the example of $D_h$ shows that 
operators exhibiting such behavior are physically relevant. 
\\
\par
Since $D_h^\delta$ \eqref{eq6} is Fredholm of index $0$, we immediately 
deduce from Theorem \ref{thm1} the following 
\begin{corollary}\label{cor1}
	Under the assumptions of Theorem \ref{thm1}, and $h>0$ small enough, 
	the spectrum of $D_h^\delta$ is discrete with probability 
	\begin{equation*}
		\geq 1- C h^{-\kappa_5}\varepsilon_0(h)
		\exp\left( - \frac{h^2}{C\varepsilon_0(h)} \log \frac{1}{\varepsilon }
		\right).
	\end{equation*} 
\end{corollary}
Consider now the perturbed chiral operator 
\begin{equation}\label{eq:Hk2}
	H_{\text{chiral}}^\delta 
	= \begin{pmatrix} 0 & D_h^\delta \\ 
		(D_h^\delta)^* & 0 
	\end{pmatrix}.
\end{equation}
In view of \eqref{eq:equiv_cond} and Corollary \ref{cor1} we get the 
the following 
\begin{thm}\label{thm2}
	Under the assumptions of Theorem \ref{thm1}, and $h>0$ small enough, 
	$H_{\text{chiral}}^\delta$ does not exhibit a flat band at energy $0$, 
	or equivalently $h>0$ small enough is not a magic angle, with probability 
	\begin{equation*}
		\geq 1- C h^{-\kappa_5}\varepsilon_0(h)
		\exp\left( - \frac{h^2}{C\varepsilon_0(h)} \log \frac{1}{\varepsilon }
		\right).
	\end{equation*} 
\end{thm}
We may take for instance $\varepsilon = \exp( -1/(Ch) -C_1 h^{-2} \varepsilon_0(h))$, 
with $C>0$ large enough, so that $h>0$ small enough is not a magic angle with 
probability $\geq 1 - C\exp(-1/(Ch))$.
\begin{figure}
\includegraphics[width=8cm]{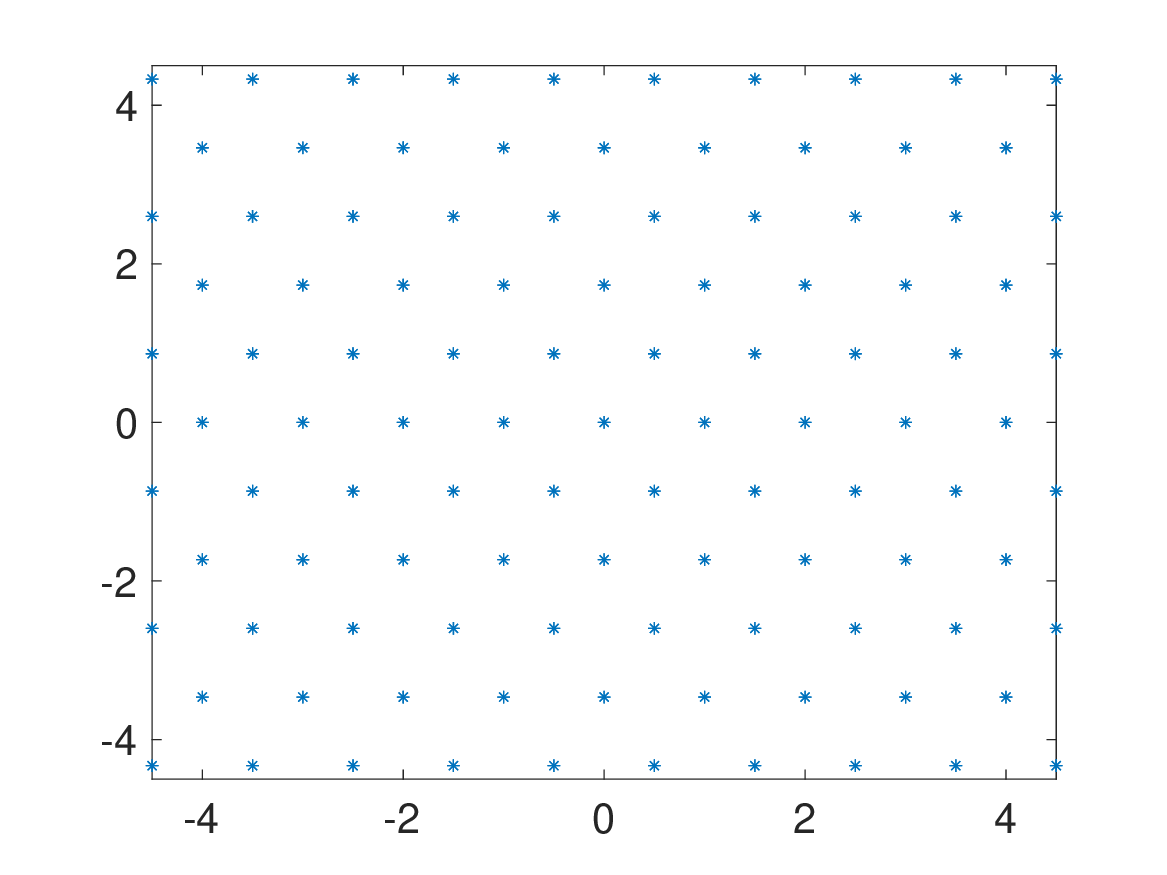} 
\includegraphics[width=8cm]{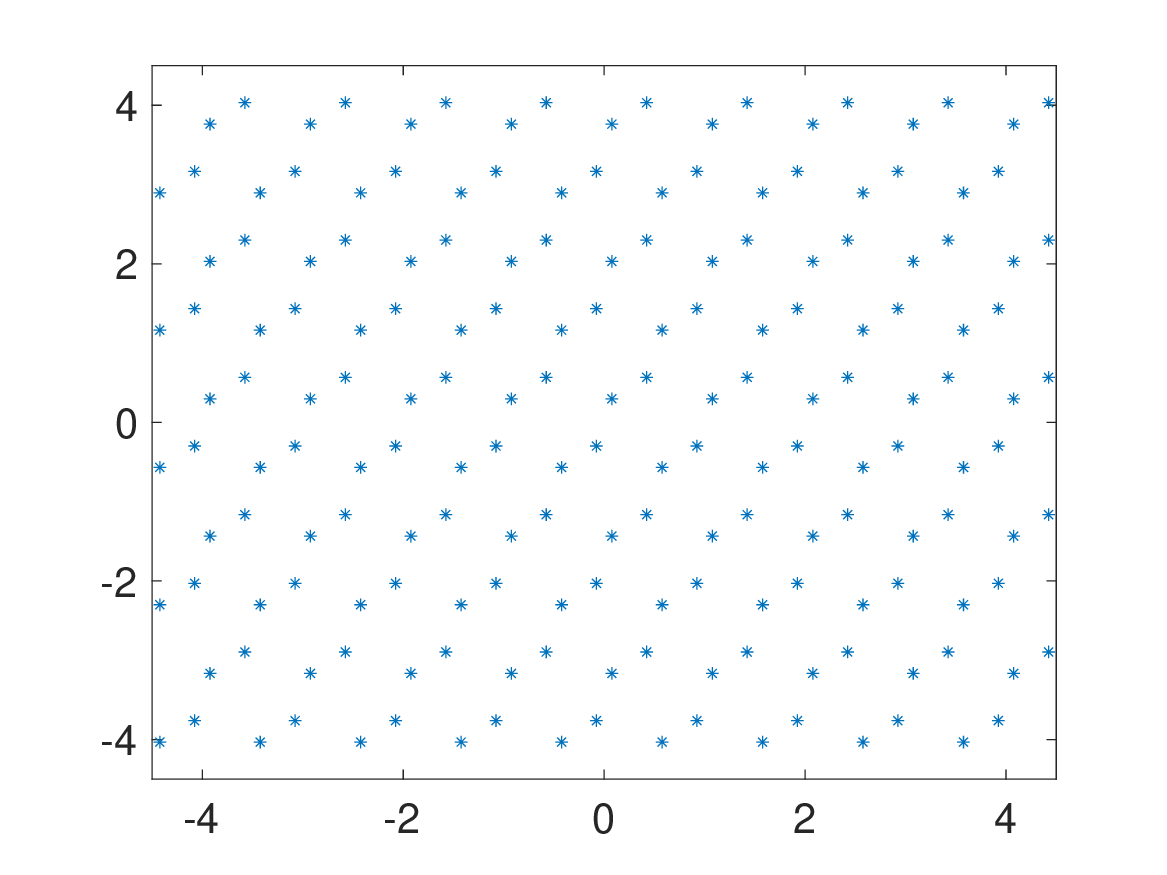} \\
\includegraphics[width=8cm]{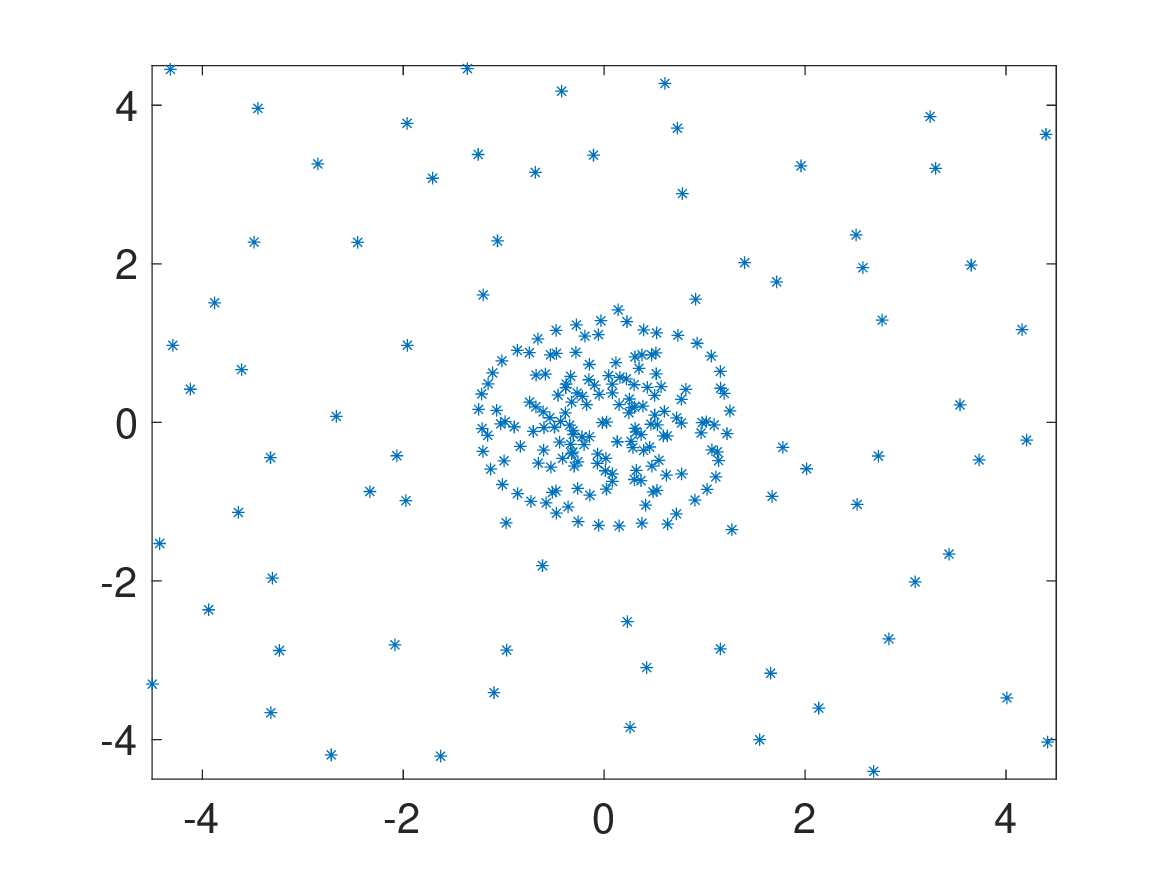} 
\includegraphics[width=8cm]{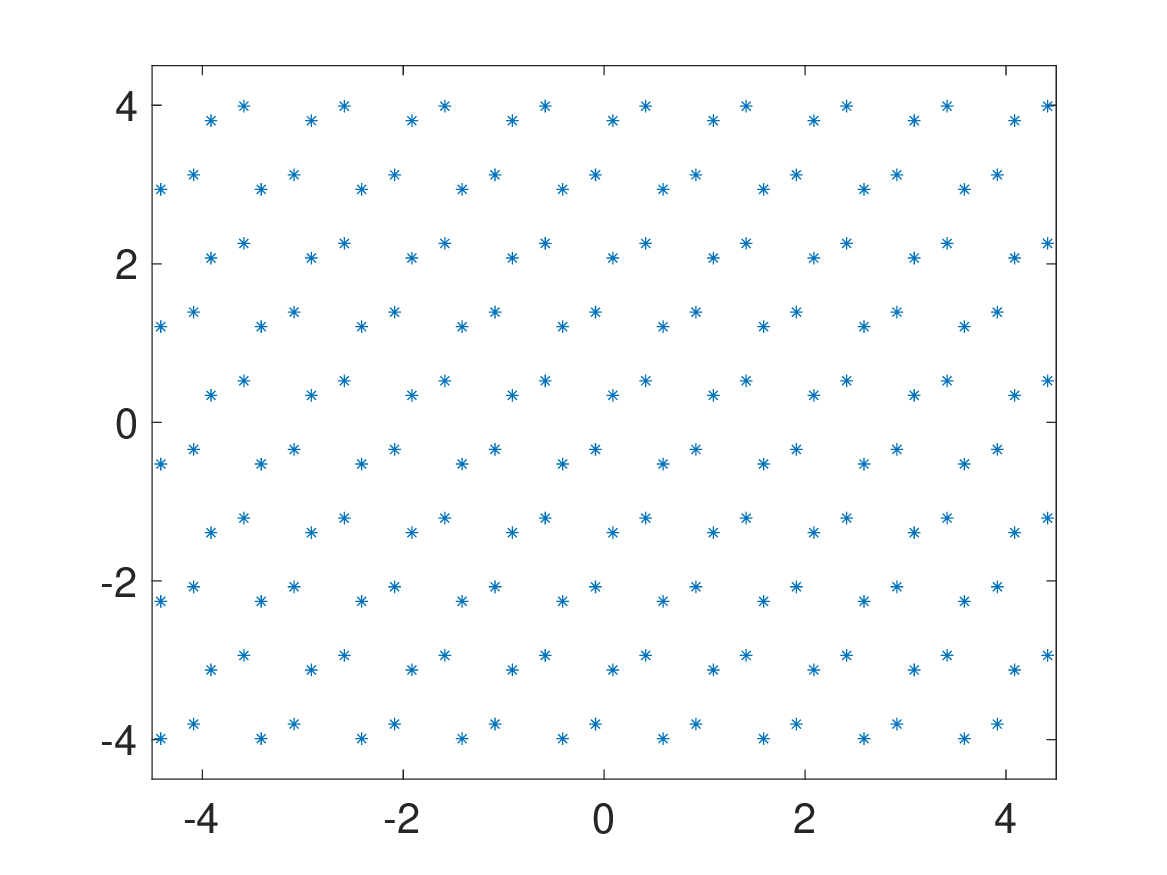}
\caption{Top left: Away from magic angles $h$, $\operatorname{Spec}(D_h)=h\Gamma^*$ 
with two-fold multiplicity.  Top right: Spectrum of randomly perturbed operator 
$D_h$ away from magic angles $h$. Bottom left: Spectrum of $D_h$ numerically 
computed at largest magic angle $h$. In reality the spectrum is the entire 
complex plane.  Bottom right: Spectrum of random perturbation of $D_h$ for 
$h$ the largest magic angle.}\label{fig1}
\end{figure}
\subsection{Eigenvalue asymptotics -- a probabilistic Weyl law}
The semiclassical principal symbol of $D_h$ is 
\begin{equation}\label{eq1.2}
	d (x,\xi)= 
	\begin{pmatrix}
		\xi_1+i\xi_2 & U(x) \\
		U(-x) & \xi_1+i\xi_2 \\
	\end{pmatrix}
	\in S^1(T^*(\C/\Gamma);\mathrm{Hom}(\C^2,\C^2)),
\end{equation}
see Appendix \ref{App:MatrixPseudo} for a definition of this symbol class.
\begin{thm}\label{thm3}
	Let $\Omega\Subset\C$ be relatively compact with Lipschitz boundary. 
	Let $\lambda_{j}(x,\xi)$, $(x,\xi)\in T^*(\C/\Gamma)$, $j=1,2$ denote the two 
	eigenvalues of $d(x,\xi)$. Then, under the assumptions of Theorem 
	\ref{thm1}, and for $h>0$ small enough,
	\begin{equation}\label{thm3:eq1}
		\#( \sigma(D_h^\delta) \cap \Omega) 
		= \frac{1}{(2\pi h)^2}\sum_1^2 \iint_{(\lambda_j)^{-1}(\Omega)} dxd\xi 
		+\mO(h^{-1})
	\end{equation}
	with probability 
	\begin{equation*}
		\geq 1- C h^{-\kappa_5}\varepsilon_0(h)
		\exp\left( - \frac{h^2}{C\varepsilon_0(h)} \log \frac{1}{\varepsilon }
		\right).
	\end{equation*} 
\end{thm}
Since the error term is independent of $\varepsilon>0$, we may take for 
instance $\varepsilon = \exp( -1/(Ch) -C_1 h^{-2} \varepsilon_0(h))$, with 
$C>0$ large enough, so that for $h>0$ small enough \eqref{thm3:eq1} holds 
with probability $\geq 1 - C\exp(-1/(Ch))$. This probabilistic Weyl law 
is numerically illustrated in Figure \ref{fig1}. 
\par
As detailed in Section \ref{sec:Count} we have that 
\begin{equation*}
	\sum_1^2 \iint_{(\lambda_j)^{-1}(\Omega)} dxd\xi 
	 =2|\Omega|\cdot |\C/\Gamma|.
\end{equation*}
\par 
In particular interest to the Theorem \ref{thm3} is the line of research by Hager 
\cite{Ha06}, Bordeaux-Montrieux \cite{BM}, Hager-Sj\"ostrand \cite{HaSj08}, 
Sj\"ostrand \cite{Sj09,Sj10a}, Christiansen-Zworski \cite{ZwChrist10}, 
and by two of the authors \cite{Ol23,Vo20} showing -- roughly speaking -- 
that for a large class of elliptic pseudo-differential operators 
and Berezin-Toeplitz quantizations the presence of a small random 
perturbation leads to a probabilistic complex version of Weyl's 
asymptotics for the eigenvalues. A lower bound on the smallest singular   
value of the perturbed operator with high probability, as in Theorem 
\ref{thm1}, is a fundamental ingredient in these proofs. 
\par 
For related results in the context of random matrix theory we refer 
the reader to the works of \'Sniday \cite{Sn02}, 
Guionnet-Wood-Zeitouni \cite{GuMaZe14,Wo16}, 
Davies-Hager \cite{DaHa09} Basak-Paquette-Zeitouni 
\cite{BPZ18,BPZ18b,BaZe20}, Bordenave-Capitaine \cite{BoCa16} and 
Sj\"ostrand together with one of the authors \cite{SjVo19b,SjVo19a}.
\subsection{Outline of the paper}
In Section \ref{sec:GPandSg} we review the construction of a well-posed 
Grushin problem and we discuss the spectral theory of the singular 
values of matrix-valued (pseudo-)differential operators. In Section 
\ref{sec:SobPert} we discuss counting estimates on the number of small 
singular values of $D_h$ and of Sobolev perturbations of $D_h$. In 
Section \ref{sec:LiftSSV} we present a modified version of a method 
developed Sj\"ostrand \cite{Sj09}, adapted to suit the operator $D_h$, 
and construct a tunneling potential $Q$ such that the smallest singular value of 
$D_h + \delta Q$ is not too small. Then we show that a good lower bound 
holds in fact with high probability. This then proves Theorem \ref{thm1}. 
Finally, in Section \ref{sec:Count}, we prove Theorem \ref{thm3}.
\subsection{Notation and conventions}
In this paper we will use the following set of notations. The notation  
$a \ll b$ means that $Ca \leq  b$ for some sufficiently large constant $C>0$. 
Writing $a \asymp b$ means that there exists a constant $C\geq1$ such that 
$C^{-1} a \leq b \leq Ca$. Writing $f = \mO(h)$ means that there exists 
a constant $C>0$ (independent of $h$) such that $|f| \leq C h$. 
When we want to emphasize that the constant $C>0$ depends on some parameter 
$\varepsilon$, then we write $C_\varepsilon$, or with the above notation 
$\mO_\varepsilon(h)$. 
\par
We use the standard notation of floor: for $x \in \R$ we write 
$\lfloor x\rfloor  := \max\{ n\in \Z; n\leq x\}$. 
\par
For $\psi\in L^2(\C/\Gamma;\C^2)$ we will frequently write 
$\psi = (\psi^1,\psi^2)$, $\psi^1,\psi^2 \in L^2(\C/\Gamma;\C)$. 
Here we equip $L^2(\C/\Gamma;\C)$ with the $L^2$ scalar 
product $(\psi^1| \phi^1) = \int_{\C /\Gamma} \psi^1 \overline{\phi^1} dx$. 
Correspondingly, we equip $\psi\in L^2(\C/\Gamma;\C^2)$ with the scalar 
product 
\begin{equation}\label{eq:SpProduct}
	(\psi| \phi) = (\psi^1|\phi^1 )+ (\psi^2|\phi^2 ), \quad 
	\psi,\phi \in L^2(\C/\Gamma;\C^2). 
\end{equation}
To avoid confusion we will denote the standard sesquilinear 
scalar product on $\C^n$, $n\geq 1$, by $\langle \lambda | \mu \rangle$ 
for $\lambda, \mu \in \C^n$. For different values of $n$ we will 
use the same notation for the scalar product whenever the context 
is clear.  
\\
\\
\textbf{Acknowledgements.} 
S. Becker acknowledges support from the SNF Grant PZ00P2 
216019. I. Oltman was jointly funded by the National Science 
Foundation Graduate Research Fellowship under grant DGE-1650114 
and by grant DMS-1901462. 
M. Vogel was partially funded by the Agence Nationale 
de la Recherche, through the project ADYCT (ANR-20-CE40-0017). 
\section{Review of Grushin problems and singular values}
\label{sec:GPandSg}
\subsection{Singular values}
For a compact operator $A:\mathcal{H}\to\mathcal{H}$ on a complex
separable Hilbert space we define, following \cite{GoKr69}, 
the \emph{singular values} of $A$ to be the decreasing sequence 
\begin{equation}\label{eq:sgv1}
	s_1(A)\geq s_2(A)\geq \dots \searrow 0,
\end{equation}
of all eigenvalues of the compact selfadjoint operator 
$(A^*A)^{1/2}$. The intertwining relations 
\begin{equation*}
	A(A^*A)=(AA^*)A, \quad (A^*A)A^* = A^*(AA^*)
\end{equation*}
imply that the non-vanishing singular values of $A$ and 
$A^*$ coincide. Furthermore, we have min-max characterization, 
see for example \cite[p. 25]{GoKr69}, of the singular values 
\begin{equation*}
	s_j(A) = \inf_{\substack{
		L\subset \mathcal{H}
	}}\sup_{u\in L\backslash\{0\}} 
		\frac{( (A^*A)^{1/2}u|u)}{(u|u)}, 
\end{equation*}
where the infimum is over all closed subspaces 
$L\subset \mathcal{H}$ of codimension $= j-1$. 
\\
\par 
The Ky Fan inequalities \cite[Corollary 2.2]{GoKr69} for two compact operators 
$A,B:\mathcal{H}\to\mathcal{H}$ give that for any $n,m\geq 1$ 
\begin{equation}\label{eq:KyFan}
	\begin{split}
		&s_{n+m-1}(A+B) \leq s_n(A)+s_m(B)\\
		&s_{n+m-1}(AB)\leq s_n(A)s_m(B).
	\end{split}
\end{equation}
\subsection{Singular values for matrix-valued differential operators}
In this section we review the construction of a well-posed Grushin 
problem for pseudo-differential operators presented in \cite[Section 17.A]{Sj19}. 
There, the theory is presented in the case of scalar-valued differential 
operators, however, it can easily be extended (by following the same arguments 
presented in \cite[Section 17.A]{Sj19} with a few modifications) to 
the case of unbounded matrix-valued operators 
\begin{equation*}
	P = P_0 + V, \quad \text{on } L^2(X;\C^n).
\end{equation*}
Here, $P_0\in\Psi_h^m(X;\C^n,\C^n)$ is a classically elliptic semiclassical $n\times n$ 
complex matrix-valued pseudo-differential operator of order $m$ on a compact smooth 
manifold $X$ equipped with a smooth density of integration, and 
$V \in L^{\infty}(X;\C^{n\times n})$. To ease the notation, we 
will write $H^s_h=H^s_h(X;\C^n)$, $s\in\R$. See Appendix \ref{App:MatrixPseudo} 
for a brief review of such operators and spaces. 

\par 
First, we view $P$ as a bounded operator $H^m_h\to H_h^0$. Then, $P^*$ 
is a bounded operator $H^0_h\to H_h^{-m}$, and the operator 
$P^*P:H^m_h\to H^{-m}_h$ satisfies 
\begin{equation*}
	(P^*Pu|u) =\| Pu\|^2 \geq 0, \quad u \in H^m_h.
\end{equation*}

Let $w\in ]0,+\infty[$. Mimicking the proof of \cite[Proposition 17.A.1]{Sj19}, 
we find that, since $P_0$ is classically elliptic, $(P^*P+w):H^m_h\to H^{-m}_h$ 
is bijective with bounded inverse. Furthermore, 
it is selfadjoint as a bounded operator $H^m_h\to H^{-m}_h$. In particular, 
when $P$ is injective, then both conclusions hold for $w=0$. 
\par 
Since the injection $H^s_h \hookrightarrow H^{s'}_h$, $s>s'$, is compact, 
the operator $(P^*P+w)^{-1}:H^{-m}_h\to H^{m}_h$ induces a compact selfadjoint 
operator $H_h^0 \to H_h^0$. Using tacitly that $P_0$ is classically elliptic 
we see that the range of $(P^*P+w)^{-1}|_{H_h^0 \to H_h^0}$ consists of all 
$u\in H^m_h$ such that $Pu\in H^m_h$. The spectral theorem for selfadjoint operators 
yields the existence of an orthonormal basis of $L^2$ composed out 
of eigenfunctions $e_1,e_2,\dots $ in $H_h^0$ such that
\begin{equation*}
	(P^*P + w)^{-1}e_j = \mu_j^2(w)e_j, \quad \mu_j(w) >0,
\end{equation*}
where $\mu_1 \geq \mu_2 \geq \dots \searrow 0$ when $j\to \infty$. 
We see by the above mapping properties that $e_j\in H^m_h$ and $Pe_j\in H^m_h$, 
so 
\begin{equation}\label{eq:genGP0.1}
	(P^*P + w)e_j = \mu_j^{-2}(w)e_j, 
\end{equation}
and 
\begin{equation*}
	P^*Pe_j = (\mu_j^{-2}(w)-w)e_j.
\end{equation*}
Since $t_j^2:=(\mu_j^{-2}(w)-w) = (P^*Pe_j|e_j)\geq 0$, we have found 
an orthonormal basis of eigenfunctions $e_1,e_2,\dots $ $\in H_h^0$ with 
$e_j\in H^m_h$ and $Pe_j\in H^m_h$ such that 
\begin{equation}\label{eq:genGP1}
	P^*Pe_j = t_j^2ej, \quad 0\leq t_j \nearrow +\infty. 
\end{equation}
On the other hand, consider $S:=P^*P:L^2 \to L^2$ as a closed unbounded 
operator with domain $\mathcal{D}(S)=\{u \in H^m_h; P^*P u \in L^2 \}$. 
Since $\mathcal{D}(S) = (P^*P+w)^{-1}L^2$, which is dense in 
$(P^*P+w)^{-1}H^{-m}_h=H^m_h$, we conclude that $S$ is densely defined.
\par 
Next, we check that $S$ is closed and self-adjoint. Indeed, if 
$(S+w)u_j=:v_j\to v$, $u_j\to u$ in $L^2$, then $v_j\to v$ in $H^{-m}_h$ 
and $u_j\to u$ in $H^{m}_h$ since $(P^*P+w)$ is a bounded bijective operator 
from $H^m_h$ to $H^{-m}_h$. In particular $(P^*P + w) u = v$, and since $v\in L^2$ 
we see that $u\in \mathcal{D}(S)$. Self-adjointness follows immediately from the 
fact that $S$ is symmetric and that the range of $S+w$ is the entire space 
$L^2$. 
\par 
Since $S$ is a positive selfadjoint operator with domain that is dense 
in $H^m_h$ which injects compactly into $L^2$, it follows that $S$ has a compact 
resolvent. So it has purely discrete spectrum and we can find an orthonormal 
basis in $L^2$ of associated eigenfunctions. In view of \eqref{eq:genGP1}, 
we find that the $t_j$ are independent of $w$ and that the orthonormal basis 
of eigenfunctions $e_j$ can be chosen independently of $w$. In particular, we 
conclude that 
\begin{equation}\label{eq:genGP2}
	\mu_j^2(w)= \frac{1}{ t_j^2 + w }. 
\end{equation}
The max-min principle \cite[Theorem XIII.2]{ReSi78} says that 
\begin{equation}\label{eq:genGP3}
t_j^2= \sup_{L\subset H_h^0}\inf_{\substack{u\in L \\ \|u\|=1}}
			(P^*Pu|u),
\end{equation}
where the supremum is taken over all closed subspaces $L$ of $H_h^0$ 
of $\mathrm{codim}\,L=j-1$, which are also in $H^m_h$. Similarly, the 
min-max principle \cite[p. 25]{GoKr69} yields
\begin{equation}\label{eq:genGP4}
	\mu_j^2(w)= \inf_{L\subset H_h^0}\sup_{\substack{u\in L \\ \|u\|=1}}
				((P^*P+w)^{-1}u|u),
\end{equation}
where the infimum is over all closed subspaces $L$ of $H_h^0$ of 
$\mathrm{codim}\,L=j-1$. When $P:H^m_h\to H_h^0$ is bijective, we can extend 
\eqref{eq:genGP4} to the case when $w=0$ and we get that 
\begin{equation}\label{eq:genGP5}
	\mu_j^2(0)=t_j^{-2}. 
\end{equation}
\par 
Assume additionally that $P:H^m_h\to H_h^0$ is a Fredholm operator of 
index $0$. The above discussion also applies to $PP^*$ where $P$ is viewed as 
a bounded operator $H^0_h\to H^{-m}_h$ and $P^*: H^m_h\to H^0_h$. We view   
$S^\dagger= PP^*$ as an unbounded operator from $L^2$ to $L^2$ with 
domain $\mathcal{D}(S^\dagger)=\{ u \in H^m_h; S^\dagger u \in L^2\}
=(PP^*+w)^{-1}L^2$, $w>0$. Analogously to the above, we can show that 
$S^\dagger$ is a selfadjoint operator with purely discrete spectrum. 
Furthermore, there exists an orthonormal basis $f_1,f_2,\dots $ in $H^0$ with 
$f_j\in H^m_h$ and $P^*f_j\in H^m_h$ such that 
\begin{equation}\label{eq:genGP6}
	PP^* f_j = \widetilde{t}_j^2 f_j, \quad 
	0\leq \widetilde{t}_j \nearrow \infty. 
\end{equation}
Following the exact same argument as in the proof of \cite[Proposition 
17.A.2]{Sj19} we get that $t_j= \widetilde{t}_j$ and can choose the 
$f_j$ so that 
\begin{equation}\label{eq:genGP7}
	Pe_j = t_j f_j, \quad 
	P^*f_j = t_j e_j.
\end{equation}
In particular, we get that the spectra of $S$ and $S^\dagger$ coincide 
\begin{equation}\label{eq:genGP7.1}
	\sigma(S) = \sigma(S^\dagger).
\end{equation}
Writing $t_j(P)= t_j$, it follows that $t_j(P^*)=\widetilde{t}_j=t_j$. In particular, 
when $P$ has a bounded inverse, then we let $s_1(P^{-1})\geq s_2(P^{-1})\geq 
\dots$ denote, in accordance with \eqref{eq:sgv1}, the singular values of $P^{-1}$ 
seen as a compact operator in $L^2$. Compactness follows from the compact injection of 
$H^m_h$ into $L^2$. Hence
\begin{equation}\label{eq:genGP8}
	s_j(P^{-1}) = \frac{1}{t_j(P)}.
\end{equation}
\subsection{Singular values and vectors of $D_h$}
\label{sec:SgValVecOfD}
In this section, we work towards setting up a well-posed Grushin problem 
for the operator $D_h$, as in \eqref{eq1}, and the perturbed operator 
\begin{equation}\label{GP:eq0.1}
	D^\delta_h:= D_h + \delta Q, \quad 0\leq \delta \ll 1,
\end{equation}
where for the moment we will only work with the assumptions 
\begin{equation}\label{GP:eq0}
	\|Q\|_{L^\infty(\C/\Gamma;\C^{2\times 2})} \leq 1, 
	\quad 
	Q= 	\begin{pmatrix}
		q^{11} & q^{12} \\
		q^{21} & q^{22} \\
	\end{pmatrix}.
\end{equation}
Later on we will strengthen our assumptions on $Q$ to be as in 
\eqref{eq5}, \eqref{eq3}, see also \eqref{eq6}. 
\\
\par 
Since $D_h:H_h^1(\C/\Gamma;\C^2) \to L^2(\C/\Gamma;\C^2)$ is Fredholm of 
index $0$, as discussed after \eqref{eq1}, the same is true for 
$D_h^\delta: H_h^1(\C/\Gamma;\C^2) \to L^2(\C/\Gamma;\C^2)$ since 
$\delta \|Q\|_{L^2\to L^2} \ll1$. 
\par
Writing $\C/\Gamma \ni x=x_1+ix_2$ we see that the semiclassical principal 
symbol of $2hD_{\overline{x}}= (hD_{x_1}+ihD_{x_2})$ is 
$\xi_1+i\xi_2$ for $(\xi_1,\xi_2)\in T_{(x_1,x_2)}^*(\C/\Gamma)$. 
So the semiclassical principal symbol of $D_h$ is given by 
\begin{equation}\label{eq1.1}
	d = 
	\begin{pmatrix}
		\xi_1+i\xi_2 & U(x) \\
		U(-x) & \xi_1+i\xi_2 \\
	\end{pmatrix}
	\in S^1(T^*(\C/\Gamma);\mathrm{Hom}(\C^2,\C^2))
\end{equation}
Here, we see $\C/\Gamma$ as a smooth compact Riemannian manifold of real 
dimension $2$ equipped with the flat Riemannian metric inherited from 
$\C\simeq \R^2$. Note that $d$ is elliptic in the classical 
sense recalled in Definition \ref{def:ClassEllip} below. 
\\
\par 
Consider the operators 
\begin{equation}\label{GP:eq1}
	S_{\delta} := (D_h^{\delta})^*D_h^{\delta}, \quad  
	\wh{S}_{\delta}:=D_h^{\delta} (D_h^{\delta})^*,
\end{equation}
as selfadjoint Friedrich's extensions from $H_h^2(\C/\Gamma;\C^2)$ 
with quadratic form domain $H_h^1(\C/\Gamma;\C^2)$. The discussion 
in the paragraph after \eqref{eq:genGP1} applies to $P= D_h^{\delta}$.  
The resulting selfadjoint operator $S=S(D_h^{\delta},)=:\widetilde{S}_{\delta}$ 
has a domain which is dense in the quadratic form domain of $S_{\delta}$. 
So the two operators coincide, i.e. $\widetilde{S}_{\delta}=S_{\delta}$. 
Similarly, we let $S^\dagger=S^\dagger((D_h^{\delta})^*)=:\widetilde{S}^\dagger_{\delta}$ 
be as in the discussion after \eqref{eq:genGP5} applied to $P^*= (D_h^{\delta})^*$ 
and we deduce that $\widetilde{S}^\dagger_{\delta}=\wh{S}_{\delta}$.
\par
By \eqref{eq:genGP7.1} we know that $S_{\delta}$ and $\wh{S}_{\delta}$ have 
the same spectrum and it contains only isolated eigenvalues of finite 
multiplicities. We denote these eigenvalues by 
\begin{equation}\label{GP:eq2}
	0\leq t_1^2 \leq \dots \leq t_N^2 \leq \dots, 
	\quad t_j(D_h^{\delta})=t_j \geq 0. 
\end{equation}
As in \eqref{eq:genGP7}, we can find two orthonormal bases of $L^2$, one 
comprised out of eigenfunctions $e_1,\dots,e_N,\dots$ of $S_{\delta}$, 
and one comprised out of eigenfunctions $f_1,\dots,f_N,\dots$ of $\wh{S}_{\delta}$, 
associated with \eqref{GP:eq2}, such that 
\begin{equation}\label{GP:eq2a}
	D_h^{\delta} e_j 
	= t_j f_j, \quad (D_h^{\delta})^* f_j 
	= t_j e_j, 
	\quad e_j,f_j \in  H_h^1(\C/\Gamma;\C^2).
\end{equation}
\par
Exploiting the following underlying symmetry of the operator $D_h$, 
we can make a different choice of eigenfunctions $f_j$. Let 
$Gu := \overline{u}$, $u\in L^2(\C/\Gamma;\C)$ and consider the 
anti-linear involution  
\begin{equation}\label{GP:eq3.2}
	\mathcal{G}:= 
	\begin{pmatrix} 
		0 & G \\ -G& 0 \\
	\end{pmatrix}, 
	\quad 
	\mathcal{G}^2 = - 1. 
\end{equation}
Assume that
\begin{equation}\label{GP:eq3.1}
	q^{11} = - q^{22} \quad \text{in \eqref{GP:eq0}.}
\end{equation}
Then,
\begin{equation}\label{GP:eq3.0}
	D_h^* = \mathcal{G} D_h \mathcal{G} \quad 
	\text{and} \quad 
	(D_h^{\delta})^* = \mathcal{G} D_h^{\delta} \mathcal{G}. 
\end{equation}
Putting 
\begin{equation}\label{GP:eq3}
	\wt{f}_j := \mathcal{G} e_j, \quad j=1,2,\dots ~,
\end{equation}
we deduce from \eqref{GP:eq3.2}, \eqref{GP:eq3.0} that 
$\wt{f}_1,\dots,\wt{f}_N,\dots$$\in \mathcal{D}(\wh{S}_{\delta})$ 
is an orthonormal basis of eigenfunctions of $\wh{S}_{\delta}$ 
associated with \eqref{GP:eq2}.
Unless all non-zero eigenvalues are simple, it is in general not 
clear that we can arrange so that the $\wt{f}_j$ satisfy \eqref{GP:eq2a}. 
\subsection{Grushin problem}\label{sec:GPgen}
In this section we will the set up a well-posed \emph{Grushin problem} 
for the unperturbed operator $D_h$ and for the perturbed operator 
$D_h^\delta$. This construction is based on the Grushin problem 
introduced in \cite{HaSj08,Sj09}.
\par 
In a nutshell, setting up a well-posed Grushin problem amounts 
to substituting the operator under investigation by an enlarged 
bijective system. This approach can be traced back to Grushin \cite{Gr}, 
where it was used to study hypoelliptic operators. In a different context 
it was employed by Sj\"ostrand \cite{Sj73}, whose notation we use. 
Grushin problems are simple yet a powerful tool in a plethora of 
subjects, such as in bifurcation theory, numerical analysis, and 
for the analysis of spectral problems. For further details we refer 
the reader to the review paper \cite{SjZw07}.
\par 
It will be crucial to be able to add another small perturbation 
of the form \eqref{eq5} to $D_h^\delta$. Therefore, we will set 
up first a well posed Grushin problem for 
\begin{equation*}
	D_h^{\delta_0} = D_h + \delta_0 Q_0
\end{equation*}
for some $0\leq \delta_0 \ll 1$ and $Q_0$ as in \eqref{GP:eq0}, 
\eqref{GP:eq3.1}, and then set up a Grushin problem for the 
perturbed operator 
\begin{equation*}
	D_h^{\delta_0} + \delta Q, 
	\quad 0\leq \delta \ll 1,
\end{equation*}
with $Q$ as in \eqref{GP:eq0}, \eqref{GP:eq3.1}.
\\
\par 
Now we turn to constructing a Grushin problem for $D_h^{\delta_0}$ while 
keeping in mind the symmetry \eqref{GP:eq3.0}. Let $t_j=t_j(D_h^{\delta_0})$ 
be as in \eqref{GP:eq2} and let $e_j$, $f_j$ be as in \eqref{GP:eq2a}. Let 
$0<\tau_0 \ll 1$ and let $N=N(\tau_0^2)\in\N$ be such that 
\begin{equation}\label{eq:sgval_cutoff}
	0\leq t_1^2 \leq \dots \leq t_N^2 \leq \tau_0^2 < t_{N+1}^2.
\end{equation}
We allow for $N=0$, in which case $\tau_0^2 < t_{1}^2$. 
Define $R_+^{\delta_0} : L^2 \to \C^N$ and $R_-^{\delta_0}:\C^N \to L^2$ by 
\begin{equation}\label{eq:sgval_cutoff.5}
	R_+^{\delta_0} u (j) := (u|e_j), 
	\quad R_-^{\delta_0}u_- = \sum_1^N u_-(j)f_j. 
\end{equation}
We define the Grushin problem 
\begin{equation}\label{GP:eq4.1}
	\mathcal{P}_{\delta_0}:= \begin{pmatrix}
			D_{\delta_0} & R_-^{\delta_0} \\ R_+^{\delta_0} & 0 \\
	\end{pmatrix}
	:
	H^1_h \times \C^N \to L^2 \times \C^N. 
\end{equation}
As in \cite{HaSj08}, the Grushin problem $\mathcal{P}_{\delta_0}$ is bijective 
with inverse 
\begin{equation}\label{GP:eq4.2}
	\mathcal{E}_{\delta_0}:= \begin{pmatrix}
			E^{\delta_0} & E_+^{\delta_0} \\ E_-^{\delta_0} & E_{-+}^{\delta_0} \\
	\end{pmatrix}
	:
	L^2 \times \C^N \to H^1_h \times \C^N, 
\end{equation}
where 
\begin{equation*}
	E^{\delta_0}v = \sum_{j=N+1}^\infty \frac{1}{t_j} (v|f_j) e_j , 
	\quad 
	E_+^{\delta_0}v_= \sum_{1}^N v_+(j)e_j, 
\end{equation*}
\begin{equation}\label{GP:eq4.3}
	E_-^{\delta_0} v(j) = (v| f_j), ~j=1,\dots,N, 
	\quad 
	E_{-+}^{\delta_0}= -\diag(t_N,\dots,t_1).
\end{equation}
Notice that 
\begin{equation}\label{GP:eq4.0}
	\| E^{\delta_0}\| \leq \frac{1}{t_{N+1}(D_h^{\delta_0})}, \quad 
	\|E_{\pm}^{\delta_0}\| = 1, \quad 
	\|E_{-+}^{\delta_0}\| = t_{N}(D_h^{\delta_0}). 
\end{equation} 
Consider now a different Grushin problem: Let $t_j$ be as in 
\eqref{GP:eq2}, let $e_j$ be as in the paragraph above 
\eqref{GP:eq2a} and let $\wt{f}_j$ be as \eqref{GP:eq3}. 
Accordingly, let $R_+^{\delta_0}$ be as in \eqref{eq:sgval_cutoff.5} 
and put 
\begin{equation}\label{GP:eqn1}
	\wt{R}_-^{\delta_0}:\C^N \to L^2, 
	\quad \quad \wt{R}_-^{\delta_0}u_- = \sum_1^N u_-(j)\wt{f}_j.
\end{equation} 
Since the two orthonormal sets of eigenfunctions $f_1,\dots,f_N$ and 
$\wt{f}_1,\dots,\wt{f}_N$ span the same space, 
it follows that there exists a unitary matrix $A\in \C^{N\times N}$ such that 
$\wt{R}_-^{\delta_0} = R_-^{\delta_0} \circ A$. Hence, the Grushin problem 
\begin{equation}\label{GP:eqn2}
	\wt{\mathcal{P}}_{\delta_0}:= \begin{pmatrix}
			D_{\delta_0} & \wt{R}_-^{\delta_0} \\ R_+^{\delta_0} & 0 \\
	\end{pmatrix}
	:
	H^1_h \times \C^N \to L^2 \times \C^N. 
\end{equation}
is bijective with inverse 
\begin{equation}\label{GP:eqn3}
	\wt{\mathcal{E}}_{\delta_0}:= 
	\begin{pmatrix}
			E^{\delta_0} & E_+^{\delta_0} \\ 
			\wt{E}_-^{\delta_0} & \wt{E}_{-+}^{\delta_0} \\
	\end{pmatrix}
	:
	L^2 \times \C^N \to H^1_h \times \C^N, 
\end{equation}
where $\wt{E}_-^{\delta_0} = A^* \circ E_-^{\delta_0}$ and 
$\wt{E}_{-+}^{\delta_0} =A^* \circ  E_{-+}^{\delta_0}$. Hence, 
the singular values $\wt{E}_{-+}^{\delta_0}$ are given by 
\begin{equation}\label{GP:eqn4}
	t_\nu ( \wt{E}_{-+}^{\delta_0}) 
	= t_\nu (E_{-+}^{\delta_0})
	= t_\nu (D_h^{\delta_0}), 
	\quad 
	\nu =1,\dots, N. 
\end{equation}
Similarly to \eqref{GP:eq4.0}, we find 

\begin{equation}\label{GP:eqn4.0}
	\| E^{\delta_0}\| \leq \frac{1}{t_{N+1}(D_h^{\delta_0})}, \quad 
	\|\wt{E}_{-}^{\delta_0}\|,\|E_{+}^{\delta_0}\|  = 1, \quad 
	\|\wt{E}_{-+}^{\delta_0}\| = t_{N}(D_h^{\delta_0}). 
\end{equation} 
\par
Next, we add an additional perturbation of the form $\delta Q$ to 
$D_h^{\delta_0}$ with $0\leq \delta \ll 1$ and $Q$ as in \eqref{GP:eq0}, 
\eqref{GP:eq3.1}. Put
\begin{equation*}
	D_h^{\delta_0,\delta} := D_h^{\delta_0} + \delta Q,
\end{equation*}
and define 
\begin{equation}\label{GP:eq4.01}
	\wt{\mathcal{P}}_{\delta_0}^\delta:= 
	\begin{pmatrix}
			D_h^{\delta_0} + \delta Q & \wt{R}_-^{\delta_0} 
			\\ R_+^{\delta_0} & 0 \\
	\end{pmatrix}
	:
	H^1_h \times \C^N \to L^2 \times \C^N. 
\end{equation}
Applying $\wt{\mathcal{E}}_{\delta_0}$ yields 
\begin{equation}\label{GP:eq4}
	\wt{\mathcal{P}}_{\delta_0}^\delta\wt{\mathcal{E}}_{\delta_0} 
	=1 
	+ 
	\begin{pmatrix}
		\delta Q E^{\delta_0} & \delta Q E^{\delta_0}_+ \\
		0 & 0 \\ 
	\end{pmatrix}
	=:1+K. 
\end{equation}
By \eqref{GP:eq0} we have $\|Q\|_{L^2\to L^2}\leq 1$. Suppose now 
that 
\begin{equation}\label{GP:eq5}
		\delta \leq t_{N+1}(D_h^{\delta_0})/2.
\end{equation}
Using a Neumann series argument, we see that the right hand 
side of \eqref{GP:eq4} is bijective with inverse bounded by 
$\|(1+K)^{-1}\| = \mO(1)$. It follows that $\wt{\mathcal{P}}_{\delta_0}^\delta$ 
is bijective with inverse given by 
\begin{equation}\label{GP:eq6}
	\wt{\mathcal{E}}_{\delta_0} ^\delta
	=\wt{\mathcal{E}}_{\delta_0}
	\left( 1 
	+ 
	\begin{pmatrix}
		\delta Q E^{\delta_0} & \delta Q E^{\delta_0}_+ \\
		0 & 0 \\ 
	\end{pmatrix}
	\right)^{-1}
	=\begin{pmatrix}
		E^{\delta_0,\delta} & E_+^{\delta_0,\delta} \\ \wt{E}_-^{\delta_0,\delta} 
		& \wt{E}_{-+}^{\delta_0,\delta} \\
\end{pmatrix}.
\end{equation}
Expanding the Neumann series, we find that 
\begin{equation}\label{GP:eq7}
	E^{\delta_0,\delta}  = 
	E^{\delta_0} + \sum_1^\infty (-\delta)^nE^{\delta_0}(QE^{\delta_0} )^n,  
\end{equation}
\begin{equation}\label{GP:eq8}
	E^{\delta_0,\delta}_+  = 
	E^{\delta_0}_+ + 
	\sum_1^\infty (-\delta)^n(E^{\delta_0}Q)^nE^{\delta_0} _+,  
\end{equation}
\begin{equation}\label{GP:eq9}
	\wt{E}^{\delta_0,\delta}_-  = 
	\wt{E}^{\delta_0}_- + 
	\sum_1^\infty (-\delta)^n\wt{E}^{\delta_0}_-(QE^{\delta_0})^n,  
\end{equation}
and 
\begin{equation}\label{GP:eq10}
	\wt{E}^{\delta_0,\delta}_{-+}  = 
	\wt{E}^{\delta_0}_{-+} -
	\delta \wt{E}^{\delta_0}_-QE^{\delta_0}_+
	+
	\sum_2^\infty (-\delta)^n 
	\wt{E}^{\delta_0}_-(QE^{\delta_0})^{n-1}QE^{\delta_0}_+. 
\end{equation}
By \eqref{GP:eq5}, \eqref{GP:eq4.0} we get 
\begin{equation}\label{GP:eq11}
	\| E^{\delta_0,\delta}  \| \leq \frac{\|E^{\delta_0} \| }{1-\delta/t_{N+1}}
	\leq \frac{2 }{t_{N+1}}, 
	\quad 
	\| \wt{E}^{\delta_0,\delta}_-  \|,\| E^{\delta_0,\delta}_+  \|
	 \leq \frac{1 }{1-\delta/t_{N+1}}	\leq 2, 
\end{equation}
and similarly 
\begin{equation}\label{GP:eq12}
	\|\wt{E}^{\delta_0,\delta}_{-+} -
	\wt{E}^{\delta_0}_{-+} 
	+
	\delta
	\wt{E}^{\delta_0}_-QE^{\delta_0}_+
	\|
	\leq 
	\frac{2\delta^2}{t_{N+1}} 
\end{equation}
It is a well-known fact of Grushin problems, see for instance 
\cite[(17.A.17),(17.A.18)]{Sj19}, that for $k=1,\dots,N$
\begin{equation}\label{GP:eq13}
	t_k(D_h^{\delta_0,\delta}) \geq 
	\frac{t_k(\wt{E}^{\delta_0,\delta}_{-+}) }
	{\|E^{\delta_0,\delta} \|t_k(\wt{E}^{\delta_0,\delta}_{-+})+ 
	\|\wt{E}^{\delta_0,\delta}_- \|\,\|E^{\delta_0,\delta}_+ \|}
\end{equation}
and 
\begin{equation}\label{GP:eq14}
	t_k(D_h^{\delta_0,\delta}) \leq \|R_-\|\,\|R_+\| 
	t_k(\wt{E}^{\delta_0,\delta}_{-+}).
\end{equation}
From \eqref{GP:eq12} we deduce that $\| t_k(E^{\delta_0,\delta}_{-+})\| 
\leq 2 t_{N+1}(D_h^{\delta_0})$. So, combining this with the estimates 
\eqref{GP:eq11}, \eqref{GP:eq12} and \eqref{GP:eq13} yields 
\begin{equation}\label{GP:eq15}
	\frac{t_k(\wt{E}^{\delta_0,\delta}_{-+})}{8} \leq 
	t_k(D_h^{\delta_0})
	\leq 
	t_k(\wt{E}^{\delta_0,\delta}_{-+}).
\end{equation}
\begin{rem}\label{rem:Zdepen}
	Note that the discussion in Section \ref{sec:SgValVecOfD} and 
	\ref{sec:GPgen} applies to $D_h$ replaced by $D_h-z$, for $z\in\C$. 
	We chose not to introduce this parameter before for ease of notation. 
\end{rem}
\section{Pseudo-differential calculus and Sobolev class perturbations}
\label{sec:SobPert}
\subsection{The unperturbed operator}
In this section we will estimate the number of small singular 
values of $D_h$, as in \eqref{eq6} with principal symbol $d$ 
as in \eqref{eq1.1}. More precisely, we wish to study the number 
of small eigenvalues of the selfadjoint operator 
\begin{equation}\label{fa:eq1}
	S = (D_h-z)^*(D_h-z), \quad z\in\Omega\Subset \C.
\end{equation}
Indeed, since the diagonal elements of $S$ are semiclassical 
Schr\"odinger operators, the Kato-Rellich theorem implies that   
$S$ is selfadjoint on $H^2(\C/\Gamma;\C^2)$. Note that  
$S=S_0$ as in \eqref{GP:eq1} (with $D_h$ replaced by $D_h-z$ as 
commented on in Remark \ref{rem:Zdepen}). 
Furthermore, $\Omega\Subset \C$ is open, connected and relatively compact. 
Notice that the principal symbol of $S$ is 
\begin{equation}\label{fa:eq2}
	s = (d-z)^*(d-z) \in S^2(T^*(\C/\Gamma);\mathrm{Hom}(\C^2,\C^2)).
\end{equation}
For $0< h \ll \alpha \ll 1$, consider the operator 
\begin{equation}\label{fa:eq3}
	\chi ( \alpha^{-1} S) = 
	-\frac{1}{\pi}\int_{\C}(\partial_{\overline{w}} \widetilde{\chi})(w)
	(w - \alpha^{-1} S)^{-1} L(dw),
\end{equation}
defined via the Helffer-Sj\"ostrand formula. Here 
$\chi \in C^{\infty}_c(\R)$, and $\widetilde{\chi}\in C^{\infty}_c(\C)$ 
is an almost holomorphic extension of $\chi$ so that 
$\widetilde{\chi}|_\R = \chi$ and $\partial_{\overline{w}} \widetilde{\chi} 
=O(|\Ima w|^\infty)$, see for instance \cite[Chapter 8]{DiSj99} or 
\cite[Chapter 14]{Zw12}. To study the trace of \eqref{fa:eq3} 
we can follow line by line the arguments of 
\cite[Section 16.3]{Sj19}, starting from formula (16.3.3) therein, 
using the matrix-valued pseudo-differential calculus reviewed in 
Appendix \ref{App:MatrixPseudo}. The only modification one needs 
to do, apart from taking care of the non-commutative nature 
of the matrix-valued calculus, is to replace the weight function 
$\Lambda$ in \cite[(16.3.4)]{Sj19} by 
\begin{equation}\label{fa:eq4}
	\Lambda := \left( 
		\frac{\alpha+\mathfrak{s}}{1+\mathfrak{s}}
		\right)^{1/2}, 
	\quad \text{where } \mathfrak{s}(x,\xi):= \tr s(x,\xi), 
	~ (x,\xi) \in T^*(\C/\Gamma).
\end{equation}
Here, the trace $\tr$ is the trace of complex $2$ by $2$ matrices. Notice that 
$\mathfrak{s}$ is the square of the Hilbert-Schmidt norm of the matrix 
$d-z$. We then get the analogue of \cite[Proposition 16.3.8]{Sj19} 
\begin{prop}\label{fa:prop1}
	Let $\chi \in C^{\infty}_c(\R;[0,1])$ and $0< h \ll \alpha \ll 1$. 
	Then for all $N,\widetilde{N}>0$, 
	\begin{equation}\label{fa:eq5}
		\| \chi(\alpha^{-1} S) \|_{\mathrm{tr}} = 
		\mO(h^{-2})(V_N(\alpha)+h^{\widetilde{N}})
	\end{equation}
	and %
	\begin{equation}\label{fa:eq6}
		\tr \chi(\alpha^{-1} S) = 
		\frac{1}{(2\pi h)^2}\iint_{T^*(\C/\Gamma)}
		\tr \chi\!\left(\frac{s(x,\xi)}{\alpha}\right)dxd\xi 
		+\mO(\alpha^{-1}h^{-1})(V_N(\alpha)+h^{\widetilde{N}}),
	\end{equation}
	uniformly in $z\in\Omega$ as in \eqref{fa:eq1}. Here,  
	\begin{equation}\label{fa:eq7}
		V_N(\alpha) = \int_0^1 \left(
			1 + \frac{t}{\alpha}
			\right)^{-N} dV(t), 
		\quad 
		V(t) = \iint_{\mathfrak{s}(x,\xi) \leq t} dx d\xi.
	\end{equation}
\end{prop}
By \eqref{eq1.1}, \eqref{fa:eq2} we find that 
\begin{equation}\label{fa:eq8}
	\mathfrak{s}(x,\xi) 
	= 2|\xi_1+ i\xi_2 - z|^2 + |U(x)|^2 + |U(-x)|^2.
\end{equation}
Since $\mathfrak{s}(x,\xi) \leq t$ implies that $2|\xi_1+ i\xi_2 - z|^2\leq t$, 
it follows that 
\begin{equation}\label{fa:eq9}
	V(t) \leq \iint_{2|\xi_1+ i\xi_2 - z|^2\leq t} dx d\xi 
	= \frac{\pi |\C/\Gamma|}{2}t. 
\end{equation}
Choosing $N>0$ sufficiently large, we see by integration by parts together 
with \eqref{fa:eq9} that 
\begin{equation}\label{fa:eq10}
	V_N(\alpha) 
	= V(1)\left(
		1 + \frac{1}{\alpha}
		\right)^{-N}  
		- \int_0^1 V(t) \frac{d}{dt}\left(
			1 + \frac{t}{\alpha}
		\right)^{-N} dt 
	=
	\mO(\alpha). 
\end{equation}
Plugging this into Proposition \ref{fa:prop1} yields 
\begin{prop}\label{fa:prop2}
	Let $\chi \in C^{\infty}_c(\R;[0,1])$ and $0< h \ll \alpha \ll 1$. 
	Then, 
	\begin{equation}\label{fa:eq11}
		\| \chi(\alpha^{-1} S) \|_{\mathrm{tr}} = 
		\mO(h^{-2}\alpha)
	\end{equation}
	and 
	\begin{equation}\label{fa:eq12}
		\tr \chi(\alpha^{-1} S) = 
		\frac{1}{(2\pi h)^2}\iint_{T^*(\C/\Gamma)}
		\tr \chi\!\left(\frac{s(x,\xi)}{\alpha}\right)dxd\xi 
		+\mO(h^{-1}),
	\end{equation}
	uniformly in $z\in\Omega$ as in \eqref{fa:eq1}. 
\end{prop}
An immediate consequence of \eqref{fa:eq12} is the desired control  
on $N(\alpha)$, the number of eigenvalues of $S$ in the interval 
$[0,\alpha]$, i.e.  
\begin{equation}\label{fa:eq13}
	N(\alpha) := \tr 1_{[0,\alpha]}(S) = \mO(h^{-2}\alpha),
\end{equation}
uniformly in $z\in\Omega$. 
\subsection{Sobolev perturbations}
Let $s>1$ and consider the perturbed operator 
\begin{equation}\label{fa:eq14}
	D^\delta_h =D_h + \delta h Q, 
	\quad 0\leq \delta \ll 1,
\end{equation}
where $Q \in H^s_h(\C/\Gamma;\C^{2\times 2})$ satisfies
\begin{equation}\label{fa:eq15}
	\| Q \|_{H^s_h(\C/\Gamma;\C^{2\times 2})} \leq 1.
\end{equation}
Thanks to Proposition \ref{app:prop4}, we know that 
$\| hQ \|_{L^\infty(\C/\Gamma;\C^{2\times 2})} \leq \mO(1)$. 
The same Proposition together with \eqref{eq:sc7} show that 
\begin{equation}\label{fa:eq16}
	^tQ, ~Q=\mO(1): H_h^{s}(\C/\Gamma;\C^{2}) 
		\to H_h^s(\C/\Gamma;\C^{2}).
\end{equation}
Here, $^tQ$ denotes the element in $H_h^{s}(\C/\Gamma;\C^{2})$ obtain by 
the matrix transpose of $Q$. 
Since $\langle Qu,\phi\rangle=\langle u,{^t}Q\phi\rangle$ for the 
$u\in H^{-s}_h$ and $\phi \in H^s_h$, we find that $^tQ, ~Q=\mO(1): 
H_h^{-s}(\C/\Gamma;\C^{2}) \to H_h^{-s}(\C/\Gamma;\C^{2})$. By 
interpolation we then get 
\begin{equation}\label{fa:eq17}
	^tQ, ~Q=\mO(1): H_h^{r}(\C/\Gamma;\C^{2}) 
		\to H_h^r(\C/\Gamma;\C^{2}), \quad -s \leq r \leq s.
\end{equation}
The same estimates and mapping properties hold true for the 
adjoint $Q^*$. 
\\
\par
For $\Omega$ as in \eqref{fa:eq1}, let $z\in\Omega$. As in \eqref{GP:eq1} 
we have 
\begin{equation}\label{fa:eq19}
	S_\delta = (D_h^\delta-z)^*(D_h^\delta-z)
	= S + \delta Q^*(D_h-z) + \delta (D_h-z)^*Q +\delta^2Q^*Q =:S+\delta R.
\end{equation}
Notice that here 
\begin{equation}\label{fa:eq19.1}
	R=\mO(1): H_h^{r+1}(\C/\Gamma;\C^{2}) \to H_h^{r-1}(\C/\Gamma;\C^{2}),
	\quad -s \leq r \leq s.
\end{equation}
First, we want to extend the estimate \eqref{fa:eq13} on the number 
of small singular values to the case of the perturbed operator $D_h^\delta$.
\begin{prop}\label{fa:prop8}
	Let $D_h$ be as in \eqref{eq6} and let $Q\in L^\infty(\C/\Gamma;\C^{2\times 2})$ with 
	$\|Q\|_{L^\infty} \leq 1$. Let $D_h^\delta = D_h + \delta Q$ with $0\leq \delta \ll1$. 
	Then, for $0<h\ll\alpha \ll1$, the number $N_\delta(\alpha)$ of eigenvalues of 
	$(D_h^\delta-z)^*(D_h^\delta-z)$ in $[0,\alpha]$ satisfies 
	\begin{equation}\label{fa:eq19.2}
		N_\delta(\alpha) = \mO(h^{-2}(\sqrt{\alpha} + \delta)^2),
	\end{equation}
	uniformly in $z\in\Omega$ as in \eqref{fa:eq1}. 
\end{prop}
\begin{proof}
By the max-min principle for selfadjoint operators, see for 
instance \cite[Theorem XIII.2]{ReSi78},	we know that the $j$-th 
eigenvalue $t_j^2$, $j=1,2,\dots$, of $(D_h^\delta-z)^*(D_h^\delta-z)$ 
is given by
\begin{equation}\label{fa:eq25.1}
	t_j^2(D_h^\delta-z) 
	= \sup_{\mathrm{codim}~L\leq j-1} 
		\inf_{\substack{u\in L\\ \|u\|=1}} 
		\| (D_h^\delta-z) u\|^2.
\end{equation}
Here, the supremum varies over subspaces $L$ of $H_h^1(\C/\Gamma;\C^{2})$ 
and the norms are with respect to $H_h^0$. Since $\| (D_h^\delta-z) u\| 
\geq \| (D_h-z) u\| - \delta \|Q\|$ for $\|u\|=1$, and since $\|Q\| \leq 1$, 
we find  
\begin{equation}\label{fa:eq25.2}
	t_j(D_h^\delta-z) \geq t_j(D_h-z) - \delta. 
\end{equation}
Therefore, when $t_j(D_h^\delta-z) \leq \sqrt{\alpha}$, then 
$t_j(D_h-z) \leq \sqrt{\alpha} + \delta$. This leads to 
\begin{equation}\label{fa:eq25.3}
	N_\delta(\alpha) \leq N((\sqrt{\alpha} + \delta)^2).
\end{equation}
The result follows from \eqref{fa:eq13}.
\end{proof}
Next, we are interested in the regularity of the eigenvectors of $S_\delta$ 
associated with its small eigenvalues. To achieve this we will follow the 
strategy of \cite[Section 16.5]{Sj19} and study the mapping properties of
\begin{equation}\label{fa:e17.1}
	\chi (  S_\delta) = 
	-\frac{1}{\pi}\int_{\C}(\partial_{\overline{w}} \widetilde{\chi})(w)
	(w -  S_\delta)^{-1} L(dw),
\end{equation}
where $\wt{\chi} \in C^{\infty}_c(\C)$ is an almost holomorphic extension of $\chi$ 
as in \eqref{fa:eq3}. Notice that \eqref{fa:e17.1} is well-defined since $S_\delta$ 
is selfadjoint, so $(w - S_\delta)^{-1}=\mO(|\Ima w|^{-1}) : 
H^0(\C/\Gamma;\C^{2})  \to H^0(\C/\Gamma;\C^{2}) $ when $w\in \C\backslash\R$. 
\\
\par
Let $W\Subset \C$ be a compact complex neighborhood of $0$, let $w\in W$, 
and let $z\in \Omega$. Since $d$ \eqref{fa:eq2} is elliptic at infinity the 
same is true for $0\leq s \in S^2(T^*(\C/\Gamma);\C^{2\times 2})$, see 
\eqref{eq1.1}. Consequently, we can find a $\psi \in C_c^\infty(T^*(\C/\Gamma);\R)$ 
and a $C>0$ such that for all $z\in \Omega$, all $w\in W$, and all 
$(x,\xi)\in T^*(\C/\Gamma)$ 
\begin{equation}\label{fa:eq18.0}
	s(x,\xi) + \psi(x,\xi)\mathbf{1}_2 - w\mathbf{1}_2
	\geq 
	\frac{
	\langle \xi \rangle ^2}{C}\mathbf{1}_2,
\end{equation}
where $\mathbf{1}_2$ is the identity map $\C^2\to \C^2$. Let 
$\Psi \in \Psi_h^{-\infty}(X;\C^2,\C^2)$ be a selfadjoint operator with principal 
symbol $\psi \mathbf{1}_2$. We can take for instance $\Psi = 
\frac{1}{2}(\Op_h(\psi\mathbf{1}_2) + \Op_h(\psi\mathbf{1}_2)^*)$, see 
Proposition \ref{app:prop5}. The left hand side of \eqref{fa:eq18.0} is 
thus elliptic in $S^2(T^*(\C/\Gamma);\C^{2\times 2})$, so for $h>0$ 
sufficiently small
\begin{equation}\label{fa:eq18}
	S + \Psi- w\mathbf{1}_2
	\geq 
	\frac{
	\langle hD_x \rangle ^2}{C}\mathbf{1}_2,
\end{equation}
for a new constant $C>0$. Recall \eqref{fa:eq19}. So, for $0<h,\delta$ 
small enough we deduce from \eqref{fa:eq18} that 
\begin{equation}\label{fa:eq20}
	S_\delta + \Psi- w\mathbf{1}_2
	\geq 
	\frac{
	\langle hD_x \rangle ^2}{2C}\mathbf{1}_2.
\end{equation}
Thus, $(S_\delta + \Psi -w )^{-1}$ exists, is uniformly bounded in $w \in W$ and 
is holomorphic in $W$. Inserting the resolvent identity 
\begin{equation}\label{fa:eq21}
	(w-S_\delta)^{-1} = (w-(S_\delta+\Psi))^{-1} - (w-S_\delta)^{-1}\Psi (w-(S_\delta+\Psi))^{-1}
\end{equation}
into \eqref{fa:e17.1} we see that the contribution from the term 
$(w-(S_\delta+\Psi))^{-1}$ is equal to zero 
by holomorphicity, provided that the support of $\widetilde{\chi}$ 
is sufficiently small. In fact, given $W\Subset\C$ a compact complex 
neighborhood of $0$ and a $\chi\in C^\infty_c(\R)$ with $\supp \chi$ 
contained in the interior of $W\cap \R$, then we can always find an 
almost holomorphic extension $\widetilde{\chi}\in C^\infty_c(\C)$ 
with $\supp \widetilde{\chi}\in W$, see for instance \cite[Chapter 8]{DiSj99}. 
So, 
\begin{equation}\label{fa:eq22}
	\chi (  S_\delta) = 
	\frac{1}{\pi}\int_{\C}(\partial_{\overline{w}} \widetilde{\chi})(w)
	(w-S_\delta)^{-1}\Psi (w-(S_\delta+\Psi))^{-1} L(dw).
\end{equation}
We can now follow line by line the arguments from (16.5.13) to (16.5.16) 
in \cite{Sj19} and conclude that for $-s\leq t\leq s$ 
\begin{equation}\label{fa:eq23}
	\begin{split}
	&(w-S_\delta)^{-1} = \mO(|\Ima w|^{-1}) : 
	H^{r-1}(\C/\Gamma;\C^{2})  \to H^{r+1}(\C/\Gamma;\C^{2}) \\
	&(w-(S_\delta+\Psi))^{-1} = \mO(1) : 
	H^{r-1}(\C/\Gamma;\C^{2})  \to H^{r+1}(\C/\Gamma;\C^{2}).
	\end{split}
\end{equation}
Since $\Psi \in \Psi_h^{-\infty}(X;\C^2,\C^2)$ we know from 
Proposition \ref{app:prop6} that $\Psi = \mO(1) : H_h^{-s+1}(\C/\Gamma;\C^{2}) 
\to H_h^{s-1}(\C/\Gamma;\C^{2})$. By \eqref{fa:eq22}, \eqref{fa:eq23} it 
follows that 
\begin{equation}\label{fa:eq24}
	\chi (  S_\delta) = \mO(1):H_h^{-s-1}(\C/\Gamma;\C^{2})
	\to H^{s+1}_h(\C/\Gamma;\C^{2}).
\end{equation}
Let $\chi \in C_c^\infty{\R;[0,1]}$ be equal to $1$ on $[0,1/2]$ with 
$\supp \chi \subset [0,1]$. Furthermore, let $e_1,e_2,\dots $ be an orthonormal 
basis of $L^2$ eigenfunctions of $S_\delta$ associated with its eigenvalues 
$0\leq t_1^2\leq t_2^2\leq \dots $, then 
\begin{equation}\label{fa:eq25}
	\chi (  S_\delta)\left( \sum_1^{N} \lambda_j e_j \right) 
	= \sum_1^{N} \lambda_j e_j. 
\end{equation}
Combining \eqref{fa:eq24} with the fact that the inclusion map 
$\iota=\mO(1): H_h^{0}(\C/\Gamma;\C^{2}) \to H_h^{-s-1}(\C/\Gamma;\C^{2})$ 
(in fact it is a compact operator) we get that $\chi (  S_\delta) = \mO(1):
H_h^{0}(\C/\Gamma;\C^{2})\to H^{s+1}_h(\C/\Gamma;\C^{2})$. Together with 
\eqref{fa:eq25} we deduce the following 
\begin{prop}\label{fa:prop7} Let $e_1,e_2,\dots $ be an orthonormal 
basis of $L^2$ eigenfunctions of $S_\delta$ associated with its eigenvalues 
$0\leq t_1^2\leq t_2^2\leq \dots $, 
	\begin{equation}\label{fa:eq26}
		\bigg\|\sum_1^{N} \lambda_j e_j \bigg\|_{H^{s+1}_h(\C/\Gamma;\C^{2})}
		\leq \mO(1)\| \lambda\|_{\ell^2}, 
	\quad 
	\lambda =(\lambda_1,\dots,\lambda_N) \in \ell^2(\{1,\dots,N\}).
	\end{equation}
\end{prop}
This is the analogue of \cite[Proposition 16.5.4]{Sj19} and 
\cite[(3.14)]{Sj09}, stated there for the scalar-valued operator 
case. 
\\
\par
Recall that each eigenfunctions $e_j$ takes values in $\C^2$ and so we 
may write $e_j=(e_j^1,e_j^2)$ for the two components. In view of 
\eqref{eq:sc7.0}, \eqref{eq:sc7} we see that \eqref{fa:prop7} implies 
\begin{equation}\label{fa:eq27}
	\bigg\|\sum_1^{N} \lambda_j e_j^k \bigg\|_{H^{s+1}_h(\C/\Gamma;\C)}
	\leq \mO(1)\| \lambda\|_{\ell^2}, 
	\quad k=1,2. 
\end{equation}
Let $k\in\{1,2\}^N$ and let $\lambda \in \ell^2(\{1,\dots,N\})$. We 
decompose $\lambda$ in to $\lambda', \lambda'' \in \ell^2(\{1,\dots,N\})$ 
as follows: for $j=1,\dots,N$, 
when $k(j) = 1$, then $\lambda'(j)=\lambda(j)$ and $\lambda''(j)=0$, 
and when $k(j) = 2$, then $\lambda'(j)=0$ and $\lambda''(j)=\lambda(j)$. 
So $\lambda = \lambda'+ \lambda ''$. Then, by \eqref{fa:eq27} and 
the triangular inequality 
\begin{equation}\label{fa:eq28}
	\begin{split}
	\bigg\|\sum_1^{N} \lambda_j e_j^{k(j)} \bigg\|_{H^{s+1}_h(\C/\Gamma;\C)}
	&\leq \bigg\|\sum_1^{N} \lambda'_j e_j^1 \bigg\|_{H^{s+1}_h(\C/\Gamma;\C)}
		+\bigg\|\sum_1^{N} \lambda''_j e_j^2 \bigg\|_{H^{s+1}_h(\C/\Gamma;\C)}
	\\
	&\leq \mO(1)(\| \lambda'\|_{\ell^2}+\| \lambda''\|_{\ell^2} )
	\\
	&\leq \mO(1)\| \lambda\|_{\ell^2}.
	\end{split}
\end{equation}
\section{Lifting the smallest singular value}
\label{sec:LiftSSV}
The aim of this section is twofold: first we construct a small tunneling 
potential $Q$ such that the smallest singular value of $D_h^\delta-z$ is 
not too small, and secondly we will show that this actually holds with 
good probability. The steps and strategy of the proof is an adaptation 
of the method developed in \cite{Sj09}.
\subsection{Matrix considerations}\label{sec:MatCon}
Let $0<\tau_0 \ll 1$ and let $N=N(\tau_0^2)$ be as in \eqref{eq:sgval_cutoff} and 
let $\{e_j\}_{j=1}^N$ be as in \eqref{GP:eq2a}. Write $e_j(x)= 
(e_j^1(x),e_j^2(x))$, $x\in \C/\Gamma$, and define the vectors
\begin{equation}\label{eq:vec1}
	\vec{e}_k(x) 
	= \begin{pmatrix}
		e_1^k(x)\\ \vdots \\ e_N^k(x) \\
	\end{pmatrix} \in\C^N, 
	\quad k =1,2. 
\end{equation}
\begin{prop}\label{prop:LinAlg}
	For every linear subspace $L\subset \C^N$ of 
	dimension $\dim L = M -1$, $M=1,\dots, N$, there exist a 
	$x_0\in \C/\Gamma$ and a $k\in \{1,2\}$ such 
	that 
	\begin{equation*}
		|\dist (\vec{e}_k(x_0),L)|^2 \geq \frac{1}{2 |\C/\Gamma| }\tr [(1-\pi_L)],
	\end{equation*}
	where $\pi_L$ denotes the orthogonal projection onto $L$. 
\end{prop}
\begin{proof}
By \eqref{eq:SpProduct} we find that for $n,m=1,\dots,N$ 
\begin{equation*}
	(e_n|e_m) = (e^1_n|e_m^1)+(e^2_n|e_m^2) =: G^1_{n,m} 
	+ G^2_{n,m}. 
\end{equation*}
We write $G^k$ for the Gramian matrices given by the coefficients 
$G^k_{n,m}$. Notice that $G:=G^1+G^2$ is the identity matrix. 
Let $\nu_1,\dots,\nu_N$ be an orthonormal basis of $\C^N$ such that 
$\nu_1,\dots,\nu_{M-1}$ is an orthonormal basis of $L$ when $M\geq 2$. 
Then, 
\begin{equation*}
\begin{split}
\int_{\C/\Gamma}|\dist (\vec{e}_k(x),L)|^2 dx 
	&= \sum_{\ell =M }^N \int_{\C/\Gamma}| (\vec{e}_k(x),\nu_\ell)|^2 dx \\
	&= \sum_{\ell =M }^N \sum_{n,m=1}^N \int_{\C/\Gamma} 
		\overline{\nu_\ell(n)}e_n^k(x) \overline{e_m^k(x)}\nu_\ell(m) dx \\ 
	&=  \sum_{\ell =M }^N (G^k \nu_\ell | \nu_\ell) \\
	&= \tr [(1-\pi_L)G^k].
\end{split}
\end{equation*}
Summing over $k$, and using that $\tr [(1-\pi_L)G^1]+\tr [(1-\pi_L)G^2]= 
\tr [(1-\pi_L)G]$, we find that 
\begin{equation*}
	\sup_{x\in\C/\Gamma}(|\dist (\vec{e}_1(x),L)|^2 + |\dist (\vec{e}_2(x),L)|^2) 
	\geq \frac{1}{|\C/\Gamma|}\tr [(1-\pi_L)]=:B. 
\end{equation*}
This implies that there exists an $x_0\in \C/\Gamma$ such that 
\begin{equation}\label{eq:LinAlgProp}
	|\dist (\vec{e}_1(x_0),L)|^2 + |\dist (\vec{e}_2(x_0),L)|^2
	\geq \frac{1}{|\C/\Gamma|}\tr [(1-\pi_L)]=:B. 
\end{equation}
Now either $|\dist (\vec{e}_1(x_0),L)|^2 \geq \frac{1}{2}B$ or 
$|\dist (\vec{e}_1(x_0),L)|^2 < \frac{1}{2}B$. In the first case 
we get the statement of the proposition for $k=1$ and in the second 
case, we deduce from \eqref{eq:LinAlgProp} that 
\begin{equation*}
	\begin{split}
	|\dist (\vec{e}_2(x_0),L)|^2 \geq \frac{1}{2}B,
	\end{split}
\end{equation*}
which concludes the proof. 
\end{proof}
Given $j\in \{1,2\}^N$ and $a = (a_1,\dots,a_N) \in (\C/\Gamma)^N$ we define 
the $N\times N$ matrix  
\begin{equation}\label{eq:matrix1}
	E := E(j,a) := (\vec{e}_{j_1}(a_1), \dots, \vec{e}_{j_N }(a_N)). 
\end{equation}
\begin{lem}\label{lem:DetLowerBound}
	There exist a $j\in \{1,2\}^N$ and $a = (a_1,\dots,a_N) \in (\C/\Gamma)^N$ 
	such that 
	\begin{equation*}
		| \det E(j,a) | \geq \frac{\sqrt{ N!}}{2^{N/2} |\C/\Gamma|^{N/2}}.
	\end{equation*}
\end{lem}
\begin{proof}
First let $L=\{0\}$. By Proposition \ref{prop:LinAlg} there exists a point 
$a_1\in\C/\Gamma$ and a $j_1\in\{1,2\}$ such that 
\begin{equation*}
	|\dist (\vec{e}_{j_1}(a_1),L)|^2 = \| \vec{e}_{j_1}(a_1)\|^2 \geq 
	\frac{1}{2 |\C/\Gamma| }.
\end{equation*}
Next, let $L=\C \vec{e}_{j_1}(a_1)$, then by Proposition \ref{prop:LinAlg} 
there exists a point $a_2\in\C/\Gamma$ and a $j_2\in\{1,2\}$ such that 
\begin{equation*}
	|\dist (\vec{e}_{j_2}(a_2),L)|^2 \geq \frac{N-1}{2 |\C/\Gamma| }.
\end{equation*}
A repeated application of Proposition \ref{prop:LinAlg} shows that for every  
$n=\{2,\dots, N\}$ there exists a point $a_n\in\C/\Gamma$ and a $j_n\in\{1,2\}$ 
such that
\begin{equation*}
	|\dist (\vec{e}_{j_n}(a_n),\C \vec{e}_{j_1}(a_1)\oplus \dots \oplus 
	\C \vec{e}_{j_{n-1}}(a_{n-1}))|^2 \geq \frac{N-n+1}{2 |\C/\Gamma| }.
\end{equation*}
Let $\nu_1,\dots,\nu_N\in\C^N$ be the Gram-Schmidt orthonormalization of 
the basis $\vec{e}_{j_1}(a_1),\dots,\vec{e}_{j_N}(a_N)$ so that  
\begin{equation*}
	\vec{e}_{j_n}(a_n) \equiv c_k \nu_k \mod(\nu_1,\dots,\nu_{n-1})
\end{equation*}
with $|c_k|^2 \geq (N-k+1)/(2|\C/\Gamma|)$. Thus, 
\begin{equation*}
	|\det E | = |\det (c_1 \nu_1, \dots , c_N \nu_N)| 
			 = | c_1 \cdot \dots \cdot c_N|,
\end{equation*}
and 
\begin{equation*}
	|\det E |^2 \geq \prod_{k=1}^N \frac{N-k+1}{2|\C/\Gamma|},
\end{equation*}
which concludes the proof. 
\end{proof}
\subsection{Admissible potential and approximation by a Dirac potential}
Analogously to \cite{Sj09,Sj10a} we use the following notion 
\begin{definition}\label{adm:Def1}
Let $L\gg 1$ and let $\psi^j_n$, $\mu^j_n$, $j=1,2$ be as in \eqref{eq2}. An 
admissible potential is a potential of the form 
\begin{equation}\label{adm:eq1}
	q^j_\alpha(x) = \sum_{0\leq \mu^j_{n}\leq L} \alpha_n \psi^j_{n}(x), 
	\quad \alpha \in \C^{D_j}, ~D_j=D_j(L)>0. 
\end{equation}
Given two admissible potentials $q^1_\alpha(x),q_\beta^2(x)$, an 
admissible tunneling potential is a potential of the form 
\begin{equation}\label{adm:eq1.2}
	Q_\gamma = 
	\begin{pmatrix}
		0 & q^1_\alpha(x) \\
		-q^2_\beta(x) & 0 \\
	\end{pmatrix},
	\quad
	\gamma=(\alpha,\beta) \in \C^{D_1}\times \C^{D_2} \simeq \C^{D}.
\end{equation}
\end{definition}
By the Weyl asymptotics \eqref{eq:sa8.00} we know 
\begin{equation}\label{adm:eq1.1}
	D_j  \asymp L^2 h^{-2}. 
\end{equation}
Such potentials can be approximated by Dirac potential in $H^{-s}$. 
The following quantitative version comes from \cite[Proposition 17.2.2]{Sj19}, 
see also \cite[Proposition 6.2]{Sj09}.
\begin{prop}\label{adm:prop1} 
	Let $s>1$, $\varepsilon\in ]0,s-1[$, and $j\in\{1,2\}$. Then there exists $C_{s,\varepsilon}>0$ 
	such that for $a\in \C/\Gamma$ there exist $\alpha_k\in \C$, $1\leq k<\infty$, such that for $L\geq 1$ 
	there exists $r\in H^{-s}(\C/\Gamma;\C)$ such that 
	\begin{equation}\label{adm:eq2}
		\delta(x-a) = \sum_{0\leq \mu_{n}^j\leq L} \alpha_n \psi^j_{n}(x) 
			+r^j(x),
	\end{equation}
	where $\|r^j\|_{H^{-s}_h} \leq C_{s,\varepsilon}L^{-(s-1-\varepsilon)}h^{-1}$, 
	and 
	\begin{equation}\label{adm:eq3}
		\left( \sum_{0\leq \mu_{n}^j\leq L}|\alpha_n|^2 \right)^{1/2} 
		\leq 
		\langle L\rangle^{1+\varepsilon}
		\left( \sum_{0\leq \mu_{n}^j\leq L}
		\langle \mu_n^j\rangle^{-2(1+\varepsilon)}|\alpha_n|^2 \right)^{1/2} 
		\leq 
		C_{s,\varepsilon}L^{1+\varepsilon}h^{-1}.
	\end{equation}
\end{prop}
Note that when the operators $P_j$ in \eqref{eq2} have real coefficients, 
then we can choose the $\alpha_k$ to be real-valued.
\subsection{A suitable Dirac tunneling potential}
Given $j\in \{1,2\}^N$ and $a = (a_1,\dots,a_N) \in (\C/\Gamma)^N$ 
we consider matrix-valued Dirac potentials of the form 
\begin{equation}\label{eq:DiracPotential}
	\widehat{Q}=\widehat{Q}(j,a)= \sum_{k=1}^N \begin{pmatrix}
		0 &  \delta_{j_k,2} \,\delta(x-a_k) \\ 
		-\delta_{j_k,1}\, \delta(x-a_k) & 0 \\
	\end{pmatrix},
\end{equation}
where $\delta(x-a)$ denotes the Dirac delta distribution at $a\in \C/\Gamma$ 
and $\delta_{n,m}$ is the Kronecker delta. 
\\
\par 
Let $0<\tau_0 \ll 1$, $N=N(\tau_0^2)\geq 1$ and let $\{e_j\}_{j=1}^N$ be as in 
the beginning of Section \ref{sec:MatCon}. Recall \eqref{GP:eq3} and suppress the 
tilde for ease of notation. Consider  
\begin{equation*}
	(\widehat{Q}(j,a) e_n | f_m) = (\widehat{Q}(j,a) e_n | \mathcal{G} e_m) 
	= \sum_{k=1}^N e_n^{j_k}(a_k)e_m^{j_k}(a_k) 
	= (E \circ E^t)_{n,m},
\end{equation*}
where in the last line we used that $E_{n,m} = e_n^{j_m}(a_m)$, 
cf. \eqref{eq:vec1} and \eqref{eq:matrix1}. 
Define the complex $N\times N$ matrix $M_{\widehat{Q}}$ by its coefficients 
\begin{equation}\label{eq:DiracPotential2.0}
	(M_{\widehat{Q}})_{n,m} := (\widehat{Q}(j,a) e_n | f_m), \quad n,m\in \{1,\dots,N\}. 
\end{equation}
Thanks to Lemma \ref{lem:DetLowerBound} there exist 
$j\in \{1,2\}^N$ and $a = (a_1,\dots,a_N) \in (\C/\Gamma)^N$ such that 
\begin{equation}\label{eq:DiracPotential2}
	|\det M_{\widehat{Q}}| = |\det E \circ E^t| \geq 
	\frac{N!}{2^{N} |\C/\Gamma|^{N}}.
\end{equation}
For $s>1$, we have that $\|\delta(\cdot-a)\|_{H_h^{-s}(\C/\Gamma;\C)}= \mO_s(1)h^{-1}$. 
Hence, we get that for all $\lambda,\mu \in\C^N$, 
\begin{equation*}
	\begin{split}
		\langle 
			M_{\widehat{Q}} \lambda |\mu 
		\rangle
		&= \sum_{n,m=1}^N (M_{\widehat{Q}})_{n,m} \lambda_m \overline{\mu}_n \\
		& = \sum_{n,m=1}^N \int \langle\widehat{Q}(j,a) e_n | \mathcal{G} e_m\rangle
			\lambda_m \overline{\mu}_n dx \\
		& = \int \sum_{k=1}^N \delta(x-a_k) \left(\sum_{n=1}^N\overline{\mu}_n 	e_n^{j_k}(x)\right)
			\left(\sum_{m=1}^N\lambda_m	e _m^{j_k}(x)\right) dx \\
		& \leq \mO_s(1)Nh^{-2} \|\mu\| \|\lambda\|,
	\end{split}
\end{equation*}
where the scalar product in the second line is with respect to the $\C^2$ structure and 
where in the last line we used \eqref{fa:eq28}. We conclude that 
\begin{equation}\label{eq:DiracPotential3}
	s_1(M_{\widehat{Q}}) 
	= \| M_{\widehat{Q}}\|_{\C^N\to\C^N} 
	\leq CNh^{-2},
\end{equation}
for some $C=C_s>0$. Since 
\begin{equation*}
	s_1(M_{\widehat{Q}})^N \geq s_1(M_{\widehat{Q}})^{k-1} s_k(M_{\widehat{Q}})^{N-k+1} 
	\geq \prod_1^N s_k(M_{\widehat{Q}}) = |\det M_{\widehat{Q}}|, 
\end{equation*}
we deduce from \eqref{eq:DiracPotential2} that 
\begin{equation}\label{eq:DiracPotential4}
	s_1(M_{\widehat{Q}}) 
	\geq 
	\frac{(N!)^{1/N}}{2|\C/\Gamma|}
	\geq\frac{N}{2\e|\C/\Gamma|},
\end{equation}
where in the last step we used Stirling's formula. Similarly, 
we get for $k\geq 2$ that 
\begin{equation*}
	s_k(M_{\widehat{Q}}) 
	\geq \left( 
		\frac{N!}{s_1(M_{\widehat{Q}})^{k-1} 2^N |\C/\Gamma|^N } 
	\right)^{\frac{1}{N-k+1}}. 
\end{equation*}
Combining this with \eqref{eq:DiracPotential3} yields 
\begin{equation}\label{eq:DiracPotential5}
	s_k(M_{\widehat{Q}}) 
	\geq 
	\frac{(N!)^{\frac{1}{N-k+1}}}
	{C^{\frac{k-1}{N-k+1}} (2|\C/\Gamma|)^{\frac{N}{N-k+1}} }
	\left(\frac{h^2}{N}\right)^{\frac{k-1}{N-k+1}}.
\end{equation}
Next, we apply Proposition \ref{adm:prop1} to approximate each 
$\delta(x-a_k)$ in \eqref{eq:DiracPotential} by an admissible potential 
as in Definition \ref{adm:Def1}, see also \eqref{eq2}, \eqref{eq3}. 
This yields an approximation of $\wh{Q}$ by an admissible tunneling potential:
\begin{equation}\label{eq:DiracPotential6}
	\widehat{Q}=Q + R, 
	\quad 
	Q =  \begin{pmatrix}
		0 &  q^1_\alpha(x) \\ 
		-q^2_\beta(x)& 0 \\
	\end{pmatrix},
	\quad 
	R=  \begin{pmatrix}
		0 &  r^1(x) \\ 
		-r^2(x)& 0 \\
	\end{pmatrix},
\end{equation} 
where 
\begin{equation}\label{eq:DiracPotential7}
	 q^1_\alpha(x) = \sum_{0\leq \mu_{n}^1\leq L} \alpha_n \psi_{n}^1(x), 
	 \quad 
	q^2_\beta(x)= \sum_{0\leq \mu_{n}^2\leq L} \beta_n \psi_{n}^2(x), 
	\quad \alpha \in \C^{D_1}, \beta \in \C^{D_2}, 
\end{equation}  
with 
\begin{equation}\label{eq:DiracPotential8}
	\|q^1_\alpha\|_{H^{-s}_h},~\|q^2_\beta\|_{H^{-s}_h}
	=\mO_{s,\varepsilon}(1)h^{-1}N, 
\end{equation}
\begin{equation}\label{eq:DiracPotential9}
	\left( \sum_{0\leq \mu_{n}^1\leq L}|\alpha_n|^2 \right)^{1/2}, 
	~
	\left( \sum_{0\leq \mu_{n}^2\leq L}|\beta_n|^2 \right)^{1/2}
	\leq 
	\mO_{s,\varepsilon}(1) L^{1+\varepsilon}h^{-1}N, 
\end{equation} 
and 
\begin{equation}\label{eq:DiracPotential10}
	\|r^k\|_{H^{-s}_h} \leq \mO_{s,\varepsilon}(1)L^{-(s-1-\varepsilon)}h^{-1}N, 
	\quad k=1,2. 
\end{equation}  
Note that the constants here do not depend on $a_k$ and $j_k$. 
\begin{prop}\label{prop:TunPotAdm} Let $s>1$ and $\varepsilon \in ]0,s-1[$. 
	We can find an admissible tunneling potential $Q$ 
	as in (\ref{eq:DiracPotential6}--\ref{eq:DiracPotential9}) such that the matrix $M_Q\in C^{N\times N}$,  
	defined by its coefficients
	\begin{equation}\label{eq:DiracPotential11}
		(M_Q)_{n,m} = (Q e_n | \mathcal{G} e_m), 
	\end{equation}  
	satisfies 
	\begin{equation}\label{eq:DiracPotential13}
		s_k(M_{Q}) 
		\geq 
		\frac{(N!)^{\frac{1}{N-k+1}}}
		{C^{\frac{k-1}{N-k+1}} (2|\C/\Gamma|)^{\frac{N}{N-k+1}} }
		\left(\frac{h^2}{N}\right)^{\frac{k-1}{N-k+1}}
		-\mO_{s,\varepsilon}(1)L^{-(s-1-\varepsilon)}h^{-2}N,
	\end{equation}
	with $C>0$ as in \eqref{eq:DiracPotential5}, and 
	\begin{equation}\label{eq:DiracPotential14}
		s_1(M_Q) \leq \mO_{s,\varepsilon}(1)h^{-2}N.
	\end{equation}
	Additionally 
	\begin{equation}\label{eq:DiracPotential14.0}
			\|q^1_\alpha \|_{H_{h}^s(\C/\Gamma;\C)}, 
			\|q^2_\beta \|_{H_{h}^s(\C/\Gamma;\C)}, 
			\|Q\|_{H_{h}^s(\C/\Gamma;\C^{2\times 2})} 
			\leq 
			\mO_{s,\varepsilon}(1)h^{-1}NL^{(s+1+\varepsilon)}.
	\end{equation}
\end{prop}
Before we turn to the proof, note that following 
the remark after Proposition \ref{adm:prop1}, we can choose the admissible 
tunneling potential obtained in Proposition \ref{eq:DiracPotential10} 
to be real-valued when the operators $P_j$ in \eqref{eq2} have real coefficients.
\begin{proof}[Proof of Proposition \ref{eq:DiracPotential10}]
	0. Recall \eqref{eq:DiracPotential6} and define the matrices 
	$M_{\widehat{Q}}, M_R\in \C^{N\times N}$ by their coefficients 
	$(M_{\widehat{Q}})_{n,m} = (\widehat{Q} e_n | \mathcal{G} e_m)$, 
	$(M_R)_{n,m} = (R e_n | \mathcal{G} e_m)$ respectively. Then, 
	\begin{equation}\label{eq:DiracPotential12}
		M_{\widehat{Q}}= M_Q+ M_R.
	\end{equation}
	1. Our first step is to control the operator norm of the matrix 
	$M_R$. So, for $\lambda,\mu\in \C^N$,  
	\begin{equation*}
		\langle 
			M_{R} \lambda |\mu 
		\rangle 
		=\int_{\C/\Gamma} \left\langle 
			R \left(\sum \overline{\mu}_n e_n(x)\right) |
			\left(\sum \lambda_m e_m(x)\right)
		\right\rangle dx.
	\end{equation*}
	By \eqref{fa:eq26}, \eqref{eq:DiracPotential10} and Proposition 
	\ref{app:prop4} 
	\begin{equation*}
	\begin{split}
		|\langle 
		M_{R} \lambda |\mu 
		\rangle|
		&\leq O_s(1) h^{-1} \|R\|_{H^{-s}_h(\C/\Gamma;\C^{2\times 2})} 
		\Big\|\left(\sum \overline{\mu}_n e_n\right)\Big\|_{H^{s}_h(\C/\Gamma;\C^{2})} 
		\Big\|\left(\sum \lambda_m e_m\right)\Big\|_{H^{s}_h(\C/\Gamma;\C^{2})}
		\\
		&\leq O_s(1)h^{-1} \|R\|_{H^{-s}_h(\C/\Gamma;\C^{2\times 2})} 
		\|\mu\|_{\ell^2} 
		\|\lambda\|_{\ell^2}
		\\
		&\leq O_{s,\varepsilon}(1)L^{-(s-1-\varepsilon)}h^{-2}N
		\|\mu\|_{\ell^2} 
		\|\lambda\|_{\ell^2}.
	\end{split}
	\end{equation*}
	Hence
	\begin{equation*}
		\| M_{R} \| \leq \mO_{s,\varepsilon}(1)L^{-(s-1-\varepsilon)}h^{-2}N,
	\end{equation*}
	which, together with \eqref{eq:DiracPotential3}, implies 
	\eqref{eq:DiracPotential14} since $L\geq1$.
	\\
	\par 
	2. Applying the first Ky Fan inequality in \eqref{eq:KyFan} to \eqref{eq:DiracPotential12} 
	yields $s_k(M_Q) \geq s_k(M_{\widehat{Q}}) - s_1(M_R)$ for $k=1,\dots,N$. The estimate 
	\eqref{eq:DiracPotential13} is then a direct consequence of \eqref{eq:DiracPotential14} 
	and \eqref{eq:DiracPotential5}. 
	\\
	\par
	3. Finally, combining \eqref{eq:sa8.0} with 
	\eqref{eq:DiracPotential8} and \eqref{eq:DiracPotential7}, we find 
	\begin{equation*}
	\begin{split}
		\|q^1_\alpha \|_{H_{h}^s(X;\C)}^2 
		&\leq 
		\mO(1)
		\sum_{\mu_n\leq L}\langle \mu_n\rangle^{2s}|\alpha_n|^2
		\\
		&\leq 
		\mO(1)L^{2(s+1+\varepsilon)}
		\sum_k\langle \mu_n\rangle^{-2(1+\varepsilon)}|\alpha_n|^2
		\\
		&\leq 
		\mO_{s,\varepsilon}(1)h^{-2}N^2L^{2(s+1+\varepsilon)}.
	\end{split}
	\end{equation*}
	Performing the same argument for $q^2_\beta$ proves 
	\eqref{eq:DiracPotential14.0}.
\end{proof}
\subsection{Lifting the smallest singular value - an iteration argument}
Fix $s>1$, $\varepsilon\in ]0, s-1[$, let $0\leq\delta_0\ll \sqrt{h}$ and put 
\begin{equation}\label{lss:eq18}
	D_h^{\delta_0} = D_h + \delta_0 Q_0, 
	\text{with } \|Q_0\|_{H_h^s(\C/\Gamma;\C^{2\times 2})} \ll h,
\end{equation}
so that $\|Q_0\|_{L^\infty(\C/\Gamma;\C^{2\times2})},\|Q_0\|_{L^2(X;\C^2)\to L^2(X;\C^2)} \leq 1$, 
see Proposition \ref{app:prop4} and \eqref{eq:DiracPotential9}. 
We suppose additionally that 
\begin{equation}\label{lss:eq18.0}
	Q_0 
	= 
	\begin{pmatrix}
		0 & q_0^1 \\ - q_0^2 & 0
	\end{pmatrix}
\end{equation}
so that $D_h^{\delta_0}$ satisfies the symmetry \eqref{GP:eq3.0}. 
\par
Fix $z\in\C$ and let $N$ denote the number of eigenvalues 
of $((D_h^{\delta_0}-z)^*(D_h^{\delta_0}-z) )^{1/2}$ in 
the interval $[0,\tau_0]$ with $\tau_0\in ]0,\sqrt{h}]$, i.e. 
\begin{equation}\label{lss:eq0}
	0\leq t_1 \leq \dots \leq t_N \leq \tau_0 < t_{N+1},
	\quad t_{\nu} = t_{\nu}(D_h^{\delta_0}-z).
\end{equation}
We allow $N=0$, in which case $t_1>\tau_0$. Setting $\alpha = Ch$, 
$C\gg 1$, we know from Proposition \ref{fa:prop8} that 
\begin{equation}\label{lss:eq0.1}
	N = N(\tau_0^2) = \mO( h^{-1} ).
\end{equation}
\begin{prop}\label{prop:p1}
1. Let $Q$ be an admissible tunneling potential as in Proposition 
\ref{prop:TunPotAdm}. Then, 
\begin{equation}\label{prop:p1.1}
	\|Q\|_{L^\infty(\C/\Gamma;\C^{2\times 2})} 
	\leq \mO(1)h^{-1} L^s \|(\alpha,\beta)\|_{\C^{D}}.
\end{equation}
2. Fix $z\in\C$, $s>1$, $\varepsilon\in ]0, s-1[$, $\eta>0$ and 
$\kappa_3\geq 2+ \frac{5(1+\varepsilon)}{s-1-\varepsilon}$. Put 
$\kappa_1= 1+ \frac{5s}{s-1-\varepsilon} + \kappa_3$, 
put $\kappa_2 = 2(\kappa_1+2) + \eta$, and let $C_L>0$ be sufficiently large. 
Then, 
\\
\\
a) when $N>0$ is large enough, there exists an admissible tunneling 
potential $Q$ as in Proposition \ref{prop:TunPotAdm} with 
\begin{equation}\label{prop:p1.2}
	L= C_Lh^{-\frac{5}{s-1-\varepsilon}}, 
\end{equation}
satisfying 
\begin{equation}\label{prop:p1.3}
	\|Q\|_{H_{h}^s(\C/\Gamma;\C^{2\times 2})}
			\leq \mO(1)h^{-\kappa_1+1},
	\quad 
	\|Q\|_{L^\infty(\C/\Gamma;\C^{2\times 2})} 
	\leq C_1 h^{-\kappa_1}
\end{equation}
for some $C_1>0$ and 
\begin{equation}\label{prop:p1.4}
	\|(\alpha,\beta)\|_{\C^{D}} = \mO(1)L^{1+\varepsilon }h^{-1}N, 
\end{equation}
such that if 
\begin{equation}\label{prop:p1.51}
	D_h^{\delta_0,\delta}= D_h^{\delta_0}+\delta \frac{h^{\kappa_1}Q}{C_1}, 
	\quad \delta = \frac{\tau_0}{C_0}h^{\kappa_1+2},
\end{equation}
for $C_0>0$ large enough, then 
\begin{equation}\label{prop:p1.5}
	\begin{split}
		t_\nu(D_h^{\delta_0,\delta} - z)
		\geq t_\nu(D_h^{\delta_0} - z) - \frac{\tau_0 h^{\kappa_1+2}}{C_0}, 
		\quad \nu >N. 
	\end{split}
\end{equation}
and for $h>0$ small enough (depending only on $\varepsilon$, $s$ and $\eta$)
\begin{equation}\label{prop:p1.6}
	t_\nu(D_h^{\delta_0,\delta})
	\geq 
	\tau_0 h^{\kappa_2}, 
	\quad 
	1 + \lfloor (1-\theta)N\rfloor  \leq \nu \leq N.
\end{equation}
b) When $N=\mO(1)$, then 2a) holds with $\nu=N$ in \eqref{prop:p1.6}.
\end{prop}
\begin{proof} The strategy of this proof is an adaptation of 
the proof of \cite[Proposition 7.2]{Sj09}. 
\\
\\
0. First notice that if $N=0$, then $t_{1}(D_h^{\delta_0}-z) > \tau_0$ and the statement 
of the Proposition follows with $Q=0$. Hence, carrying on, we may assume 
that $N\geq 1$.
\\
\\
1. Let $Q$ be an admissible tunneling potential. 
By \eqref{eq:sa8.0} we have that for $j=1,2$ 
\begin{equation*}
	\begin{split}
		\|q^j_\alpha \|_{H_{h}^s(\C/\Gamma;\C)}^2 
		&\leq 
		\mO(1)
		\sum_{\mu_n\leq L}\langle \mu_n^j\rangle^{2s}|\alpha_n|^2
		\\
		&\leq \mO(1)L^{2}\|\alpha\|_{\C^{D_j}}^2.
	\end{split}
 \end{equation*}
Combining Proposition \ref{app:prop4} and \eqref{eq:sa8.1} we find
\begin{equation}\label{lss:eq5}
	\|Q\|_{L^\infty(\C/\Gamma;\C^{2\times 2})} 
	\leq \mO(1)h^{-1}\|Q\|_{H_{h}^s(\C/\Gamma;\C^{2\times 2})}
	\leq \mO(1)h^{-1} L^s \|(\alpha,\beta)\|_{\C^{D}},
\end{equation}
and we conclude \eqref{prop:p1.1}.
\\
\par
If additionally $Q$ is an admissible tunneling potential as in 
Proposition \ref{prop:TunPotAdm}. Then, by \eqref{eq:DiracPotential9}, 
\begin{equation}\label{lss:eq6}
	\|(\alpha,\beta)\|_{\C^{D}} = \mO(1)L^{1+\varepsilon }h^{-1}N. 
\end{equation}
Notice that the constants in the estimates only depend on $s,\varepsilon$ 
(apart from global constants) which are fixed for our current purposes. In 
fact this is true for all constants in the what follows and will be used 
tacitly without being mentioned explicitly.  
\\
\par 
2. Next, we set up a Grushin Problem for $D_h^{\delta_0}$ as in 
\eqref{eq1}. Proceeding as in Section \ref{sec:GPgen} while 
keeping in mind the symmetry \eqref{GP:eq3.0} and the choice \eqref{GP:eq3} 
(we suppress the ``$\sim$'' superscript for ease of notation), to obtain 
a bijective Grushin problem 
\begin{equation}\label{lss:eq1}
	\mathcal{P}(z) = \begin{pmatrix} 
	D_h^{\delta_0}-z  & R_-^{\delta_0} \\ 
	R_+^{\delta_0} & 0\\
	\end{pmatrix} :  H^1 \times \C^N \longrightarrow L^2 \times \C^N,
\end{equation}
with 
\begin{equation}\label{lss:eq1.1}
	R_+^{\delta_0} =   \sum_1^N \delta_j \circ e_j^*, \quad 
	R_-^{\delta_0}=\sum_1^N  f_j \circ \delta_j^*,
\end{equation}
and $f_j = \mathcal{G}e_j$, see \eqref{GP:eq3}. Its inverse is given by 
$\mathcal{E}$ of the form \eqref{GP:eqn3}, satisfying the estimates 
\eqref{GP:eqn4.0}, which by \eqref{lss:eq0} become 
\begin{equation}\label{lss:eq2}
\|E^{\delta_0}\| \leq \tau_0^{-1}, \quad \|E_{\pm}^{\delta_0}\| =1, 
\quad \| E_{-+}^{\delta_0}\|  \leq \tau_0.
\end{equation}
By \eqref{GP:eqn4} we have  
\begin{equation}\label{lss:eq3}
t_{\nu}(E_{-+}^{\delta_0}(z)) = t_{\nu}(D_h^{\delta_0}-z), \quad \nu =1,\dots, N.
\end{equation}
Here, on the left hand side, $t_{\nu}(E_{-+}^{\delta_0}(z))$ are the eigenvalues 
of $((E_{-+}^{\delta_0}(z))^*E_{-+}^{\delta_0}(z))^{1/2}$ ordered in an increasing way and 
counting multiplicities, similarly to \eqref{lss:eq0}. 
\\
\\
3. Assume that $N\gg 1$. We will distinguish two cases: First, suppose that 
\begin{equation}\label{lss:eq3.0}
	s_j(E_{-+}^{\delta_0}(z)) \geq  \tau_0 h^{\kappa_2}, 
	\quad \text{for all } j=1,\dots, N - \lfloor (1-\theta)N\rfloor .
\end{equation}
Setting $Q=0$ we have $D_h^{\delta} = D_h^{\delta_0}$ and 
\begin{equation}\label{lss:eq4}
t_j(D_h^{\delta,\delta_0}-z) 
= s_{N-j+1}(E_{-+}^{\delta_0}(z)) \geq  \tau_0 h^{\kappa_2}, \quad 
\lfloor (1-\theta)N\rfloor  +1 \leq j \leq N. 
\end{equation}
Secondly, we consider the case when 
\begin{equation}\label{lss:eq3.1}
s_j(E_{-+}^{\delta_0}(z)) <  \tau_0 h^{\kappa_2} 
\text{ for some } j=1,\dots, N - \lfloor (1-\theta)N\rfloor .
\end{equation}
Next, $Q$ is an admissible tunneling potential as in 
Proposition \ref{prop:TunPotAdm} and choose $L$ large enough, so that 
the first term on the right hand side of \eqref{eq:DiracPotential13} 
is dominant. Combining \eqref{eq:DiracPotential13} and Stirling's formula 
we find that for some constant $C>0$ and for all $1\leq k \leq N/2$ 
\begin{equation}\label{lss:eq7}
	s_k(M_{Q}) 
	\geq 
	\frac{h^2}{C}
	-\mO(1)L^{-(s-1-\varepsilon)}h^{-2}N.
\end{equation}
In view of \eqref{lss:eq0.1}, taking $L$ as in \eqref{prop:p1.2} with 
$C_L>0$ sufficiently large, we conclude that 
\begin{equation}\label{lss:eq8}
	s_k(M_{Q}) 
	\geq 
	\frac{h^2}{2C}, \quad 
	\text{for } 1\leq k \leq N/2. 
\end{equation}
Additionally, with this choice of $L$, we can combine \eqref{lss:eq5}, 
\eqref{lss:eq6} to deduce \eqref{prop:p1.3} and \eqref{prop:p1.4}.
\\
\par 
Continuing, put $D_h^{\delta} = D_h^{\delta_0} + \delta \wt{Q}$ with 
\begin{equation*}
	\wt{Q} := C_1^{-1}h^{\kappa_1} Q,
\end{equation*}
so that in view of \eqref{prop:p1.3}
\begin{equation}\label{lss:eq6.1}
	\|\wt{Q}\|_{L^2(X;\C^2)\to L^2(X;\C^2)}\leq 
	\|\wt{Q}\|_{L^\infty(X;\C^{2\times 2})} \leq 1. 
\end{equation}
If $\delta\leq \tau_0/2$, the construction of the 
perturbed Grushin problem (\ref{GP:eq4.01}-\ref{GP:eq15}) 
goes through. Since $E_-^{\delta_0} \wt{Q}E_+^{\delta_0} 
= C^{-1}h^{\kappa_1}M_Q$, we get using \eqref{lss:eq8} 
\begin{equation}\label{lss:eq7.1}
	s_k(\delta E_-^{\delta_0} \wt{Q}E_+^{\delta_0} ) 
	= C^{-1}\delta h^{\kappa_1} s_k(M_{Q}) 
	\geq \frac{\delta h^{2+\kappa_1}}{C}, \quad 
	\text{for } 1\leq k \leq N/2.
\end{equation}
Here, the constant $C>0$ changes in the last inequality.  
\par 
By the Ky Fan inequalities \eqref{eq:KyFan}, we get that for three 
compact operators $A,B,C$ 
\begin{equation*}
	s_n(A+B+C)
	\geq 
	s_{n+j+\ell-2}(A)
	-s_j(B)
	-s_\ell(C).
\end{equation*}
Applying this to \eqref{GP:eq10} with $\ell=1$ yields 
by \eqref{GP:eq12} 
\begin{equation}\label{lss:eq9}
	s_n(E^{\delta_0,\delta}_{-+})
	\geq 
	s_{n+j-1}(\delta E_-^{\delta_0} \wt{Q}E_+^{\delta_0} )
	-s_j(E^{\delta_0,\delta}_{-+})
	-\frac{2\delta^2}{\tau_0}.
\end{equation}
Put $j=N-\lfloor (1-\theta)N\rfloor $. Since $N\gg 1$ and 
$\theta \in ]0,1/4[$, we find that $1\leq n+j-1 \leq N/2$. 
So, using \eqref{lss:eq3.1} with $j=N-\lfloor (1-\theta)N\rfloor $ and 
\eqref{lss:eq7.1} gives 
\begin{equation}\label{lss:eq10}
	s_n(E^{\delta_0,\delta}_{-+})
	\geq 
	\frac{\delta h^{2+\kappa_1}}{C}
	-\tau_0h^{\kappa_2}
	-\frac{2\delta^2}{t_0}, 
	\quad 
	1 \leq n \leq N-\lfloor (1-\theta)N\rfloor .
\end{equation}
Put
\begin{equation}\label{lss:eq11}
	\delta = \tau_0 h^{\kappa_1+2}/C_0, 
\end{equation}
with $C_0>0$ large enough. Then for $h>0$ small enough
\begin{equation}\label{lss:eq12}
	s_n(E^{\delta_0,\delta}_{-+})
	\geq 
	\frac{\delta h^{2+\kappa_1}}{C}, 
	\quad 
	1 \leq n \leq N-\lfloor (1-\theta)N\rfloor .
\end{equation} 
for some $C>0$. Hence, for $h>0$ small enough 
\begin{equation}\label{lss:eq13}
	s_n(E^{\delta_0,\delta}_{-+})
	\geq 
	\tau_0 h^{\kappa_2}, 
	\quad 
	1 \leq n \leq N-\lfloor (1-\theta)N\rfloor ,
\end{equation}
and we conclude \eqref{prop:p1.6}. 
\par 
Finally, the min-max principle in combination 
with \eqref{lss:eq5}, \eqref{lss:eq6.1} and \eqref{lss:eq11} yields 
\begin{equation}\label{lss:eq5.1}
\begin{split}
	t_\nu(D_h^{\delta,\delta_0} - z)
	 &\geq t_\nu(D_h^{\delta_0} - z) - \frac{\tau_0 h^{\kappa_1+2}}{C_0},  
	\quad \nu >N. 
\end{split}
\end{equation}
4. When $N=\mO(1)$, we note that \eqref{lss:eq10} still holds with 
$\nu = N$ which yields the last statement of the proposition. 
\end{proof}
A key observation is that the construction in Proposition \ref{prop:p1} 
can be iterated. More precisely we may plug this Proposition into 
the iterative scheme \cite[(7.31)-(7.40)]{Sj09} to deduce directly 
Proposition \ref{prop:p2} below. Since this scheme is very beautiful we 
present it here for the reader's convenience: 
\\
\par  
First, suppose that $N\gg1$ is sufficiently large. We begin with 
$(D^{(0)},N^{(0)},\tau_0^{(0)}):=(D_h,N,\tau_0)$ and 
obtain by part 2 of Proposition \ref{prop:p1} a triple 
($D_h^{\delta},\lfloor (1-\theta)N\rfloor ,\tau_0 h^{\kappa_2})
=: (D^{(1)},N^{(1)},\tau_0^{(1)})$ with $D_h^{\delta} = D_h + 
\delta^{(1)} h^{\kappa_1} Q^{(1)}/C_1$, $\delta^{(1)} = \delta=\tau_0^{(0)}h^{\kappa_1+2}/C_0$ 
and $Q^{(1)}$ an admissible tunneling potential satisfying 
(\ref{prop:p1.2} - \ref{prop:p1.4}) such that 
%
\begin{equation}\label{lss:eq20}
\begin{split}
&t_{\nu}(D^{(0)}-z) 
	\geq \tau_0^{(0)}, \quad N^{(0)} +1 \leq \nu, \\
&t_{\nu}(D^{(1)}-z) 
\geq \tau_0^{(1)}, \quad N^{(1)} +1 \leq \nu \leq N^{(0)}, \\
&t_{\nu}(D^{(1)}-z) \geq  t_{\nu}(D^{(0)}-z)  
	- \delta^{(1)}, \quad \nu > N^{(0)}. \\
\end{split}
\end{equation}
Notice here that $N^{(1)}$ is the number of singular values of 
$D^{(1)}$ smaller or equal than $\tau_0^{(1)}$. 
\par 
In the next step, we replace $(D^{(0)},N^{(0)},\tau_0^{(0)})$ 
by $(D^{(1)},N^{(1)},\tau_0^{(1)})$ and apply again Proposition \ref{prop:p1}. 
Continuing on like this, we get a sequence $(D^{(k)},N^{(k)},\tau_0^{(k)})$, 
for $k=0,1,\dots, k(N)$, given by 
\begin{equation}\label{lss:eq21}
	\begin{split}
	&D^{(k+1)} = D^{(k)} + \delta^{(k+1)}\frac{h^{\kappa_1}}{C_1}Q^{(k+1)} 
	, \quad \delta^{(k+1)} = \tau_0^{(k)} h^{\kappa_1+2}/C_0, 
	\quad \tau_0^{(k)}=\tau_0 h^{\kappa_2 k}, \\ 
	& N^{(k+1)} = [(1-\theta)N^{(k)}],\quad 
	D^{(0)} = D_h, \quad N^{(0)} = N, \quad \delta^{(1)}=\tau_0\frac{h^{\kappa_1+2}}{C_0} .
	\end{split}
\end{equation}
Here, $k(N)$ is determined by the condition that $N^{(k(N))}$ is of the order 
of magnitude of a large constant and $Q^{(k)}$, $k=1,\dots k(N)$, are admissible 
tunneling potentials satisfying (\ref{prop:p1.2} - \ref{prop:p1.4}). Furthermore, 
\begin{equation}\label{lss:eq22}
\begin{split}
&t_{\nu}(D^{(k+1)}-z) \geq  t_{\nu}(D^{(k)}-z)  - \delta^{(k+1)}, 
\quad \nu > N^{(k)}, \\
&t_{\nu}(D^{(k)}-z) \geq \tau_0^{(k)}, \quad N^{(k)} + 1 \leq \nu \leq N^{(k-1)}
\\
&t_{\nu}(D^{(k+1)}-z) \geq \tau_0^{(k+1)}, \quad N^{(k+1)} + 1 \leq \nu \leq N^{(k)}.
\end{split}
\end{equation}
%
%
Since $N^{(k)} \leq (1-\theta)^{k}N$, we see that the condition on $k(N)$, 
i.e. that $k(N)$ is the first index such that $N^{(k(N))} \leq C$ for some sufficiently 
large $C>0$, implies that $C\geq N^{(k(N))}>(1-\theta)C$ and 
\begin{equation}\label{lss:eq23}
k \leq \frac{\log\frac{N}{C}}{\log\frac{1}{1-\theta}}.
\end{equation}
For $\nu > N$, we apply the first estimate in \eqref{lss:eq22} iteratively 
and get 
\begin{equation}\label{lss:eq24}
\begin{split}
t_{\nu}(D^{(k+1)}-z) &\geq  t_{\nu}(D_h-z)  -  \frac{h^{\kappa_1+2}\tau_0}{C_0} \sum_0^k h^{\kappa_2 k}, \\
& \geq  t_{\nu}(D_h-z)  - \tau_0 h^{\kappa_1+2}\frac{1}{C_0(1-h^{\kappa_2})}.
\end{split}
\end{equation}
For $1 \ll \nu \leq N$, we let $\ell = \ell(N)$ be the unique index for which 
$N^{(\ell)}+1\leq \nu \leq N^{(\ell -1)}$. The second inequality in 
\eqref{lss:eq22} implies 
\begin{equation}\label{lss:eq25}
t_{\nu}(D^{(\ell)}-z) 
\geq \tau_0^{(\ell)} =\tau_0 h^{\kappa_2\ell}.
\end{equation}
For $k \geq \ell $, the first estimate in \eqref{lss:eq22} together with 
\eqref{lss:eq25} gives
\begin{equation}\label{lss:eq26}
\begin{split}
t_{\nu}(D^{(k+1)}-z) &\geq 
\tau_0 h^{\kappa_2\ell}  -  \frac{h^{\kappa_1+2}\tau_0}{C_0} \sum_\ell^k h^{\kappa_2 k}, \\
& \geq  \tau_0 h^{\kappa_2\ell}   - \tau_0 h^{\kappa_1+2 + \kappa_2\ell}\frac{1}{C_0(1-h^{\kappa_2})} \\
& \geq \tau_0 h^{\kappa_2\ell}\left(1 - h^{\kappa_1+2}\frac{1}{C_0(1-h^{\kappa_2})} \right).
\end{split}
\end{equation}
This iteration works until $k = k(N) = \mO(\frac{\log\frac{N}{C}}{\log\frac{1}{1-\theta}})$, in which case 
$N^{(k(N))} = \mO(1)$. For $k > k(N)$ we continue iterating, by decreasing $N^{(k)}$ by $1$ at each 
step until we reach $N^{(k_0)} =1$, with 
\begin{equation}\label{lss:eq27}
k_0 = k(M) + \mO(1) = \mO(\log N) + \mO(1) = \mO(\log N).
\end{equation}
By part 3 of Proposition \ref{prop:p1} we see that \eqref{lss:eq22}, 
and consequently \eqref{lss:eq25} and \eqref{lss:eq26}, still hold. 
In particular, we see that after our last step we end up with 
\begin{equation}\label{lss:eq28}
	t_{1}(D^{(k_0+1)}-z) \geq \tau_0 h^{\kappa_2(k_0+1) } 
	=\tau_0 h^{\kappa_2 O(\log N)}.
\end{equation}
By construction
\begin{equation}\label{lss:eq29}
\begin{split}
	D^{(k_0+1)} 
		&= D_h + \sum_0^{k_0}\delta^{(k+1)}\frac{h^{\kappa_1}}{C_1}Q^{(k+1)} \\
		&=D_h + \delta \sum_0^{k_0}\frac{h^{\kappa_1+\kappa_2k}}{C_1}Q^{(k+1)} \\ 
		&=:D_h +\delta \frac{h^{\kappa_1}}{C_1}Q 
\end{split}
\end{equation}
where $\delta=\tau_0\frac{h^{\kappa_1+2}}{C_0}$ and 
all $Q^{(k)}$ are admissible tunneling potentials satisfying 
(\ref{prop:p1.2} - \ref{prop:p1.4}) uniformly. Therefore, 
$Q$ is admissible tunneling potentials satisfying 
(\ref{prop:p1.2} - \ref{prop:p1.4}). 
\\
\par 
In conclusion we get the analogue of \cite[Proposition 7.3]{Sj09}: 
\begin{prop}\label{prop:p2}
	Let $z\in\C$ be fixed and let $D_h$ be as in \eqref{eq1}. 
	Let $s>1$, $\varepsilon\in ]0, s-1[$, $\eta>0$ and fix 
	$\kappa_3\geq 2+ \frac{5(1+\varepsilon)}{s-1-\varepsilon}$. 
	Put $\kappa_1= 1+ \frac{5s}{s-1-\varepsilon} + \kappa_3$, put 
	$\kappa_2 = 2(\kappa_1+2) + \eta$, and let $C_L>0$ be sufficiently 
	large. Let $\tau_0 \in ]0,\sqrt{h}]$ and fix $\theta \in ]0,1/4[$ and 
	let $N_\theta \gg 1$. Define the finite decreasing sequence of integers  
	$N^{(k)}$ by putting $N^{(k+1)}=[(1-\theta)N^{(k)}]$ for $0\leq k \leq 
	k(N_\theta)$, where $k(N_\theta)$ is the last index for which 
	$N^{(k)}\geq N_\theta$. Then, we define $N^{(k+1)} = N^{(k)}-1$ for 
	$k(N_\theta)\leq k \leq k_0$, where $k_0$ is such that $N^{(k_0)}=1$.
	\par 
	Then, there exists an admissible tunneling potential $Q$ as in \eqref{eq5} 
	satisfying \eqref{prop:p1.1}, \eqref{prop:p1.3} and \eqref{prop:p1.4} 
	such that if $D_h^\delta = D_h + \delta h^{\kappa_1} Q/C_1$, 
	$\delta =\tau_0h^{\kappa_1+2}/C_0$, $C_0>0$ large enough, we have
	\begin{itemize}
		\item If $\nu > N^{(0)}$, then $t_\nu(D_h^\delta -z)\geq 
				\tau_0(1 - \mO(h^{\kappa_1+2}))$. 
		\item If $N^{(k)}+1\leq \nu \leq N^{(k-1)}$, $k=1,\dots, k_0$, 
		then $t_\nu(D_h^\delta -z)\geq \tau_0 h^{\kappa_2 k}(1 - \mO(h^{\kappa_1+2}))$, 
		where the constant in the error term is uniform in $k$ and $\nu$, 
		\item and $t_1(D_h^\delta -z)\geq \tau_0 h^{\kappa_2 (k_0+1)}$ with 
		$k_0 = \mO(\log h^{-1})$
	\end{itemize} 
\end{prop}
Note that, following the remark after Proposition \ref{adm:prop1}, 
the admissible tunneling potential $Q$ obtained in Proposition 
\ref{prop:p2} can be chosen to be real-valued if the operators 
$P_j$ in \eqref{eq2} have real coefficients. 
\subsection{Probabilistic lower bound on the smallest singular value}
Let $z\in\C$ be fixed, let $s>1$, $\varepsilon\in ]0, s-1[$ and 
let $\kappa_1,\kappa_2,\kappa_3$ be as in Proposition \ref{prop:p2}. 
For $C_L>0$ large enough let $L$ be as in \eqref{prop:p1.2}. Assume 
\begin{equation}\label{plb:eq1}
	Ch^{-2-\frac{5s}{s-1-\varepsilon}} \leq R \leq Ch^{-\kappa_3},
\end{equation}
for $C>0$ large enough. Consider $D_h$ as in \eqref{eq1} and $Q=Q_\gamma$ 
as in \eqref{eq5} with 
\begin{equation}\label{plb:eq1.0}
	\|\gamma\|_{\C^{D}} \leq 2R, \quad \gamma = (\alpha,\beta) 
	\in \C^{D_1}\times \C^{D_2} \simeq \C^{D}.
\end{equation}
For $\delta = \tau_0 h^{\kappa_1+2}/C_0$, $C_0>0$ large enough, put 
\begin{equation}\label{plb:eq1.1}
	D^\delta_h = D_h + \delta h^{\kappa_1}Q.
\end{equation}
Using \eqref{prop:p1.1}, \eqref{prop:p1.2}, the Grushin problem \eqref{GP:eq4.01} 
applied to \eqref{plb:eq1.1} with $\delta_0=0$ and $N=N(\tau_0^2)=\mO(h^{-1})$, 
is bijective with bounded inverse \eqref{GP:eq6}. The Shur 
complement formula shows that $D_h^\delta-z$ is bijective if and only if 
$E_{-+}^\delta(z)$ is bijective. Hence, $z$ is in the spectrum of $D_h^\delta$ 
if and only if 
\begin{equation}
	\det E_{-+}^\delta(z) = 0, \quad E_{-+}^\delta(z)=E_{-+}^\delta(z,Q) \in \C^{N\times N}.
\end{equation}
Since $Q$ depends holomorphically on $\gamma$, so does $\det E_{-+}^\delta(z,Q)$, by \eqref{GP:eq10} and \eqref{GP:eq12}. 
By Proposition \ref{prop:p2} and \eqref{GP:eq15} we can perform the 
same calculation as in \cite[Formula (7.46)]{Sj09} and conclude that 
there exists a constant $C>0$ such that for $Q$ as in Proposition \ref{prop:p2} 
\begin{equation}\label{plb:eq2}
\begin{split}
	\log |\det E_{-+}^{\delta} | &\geq\sum_1^N \log t_\nu(D_h^\delta -z) \\ 
		&\geq -C (\log \tau_0^{-1} + (\log h^{-1})^2)(h^{-1}+\log h^{-1}).
\end{split}
\end{equation}
We want to show that we have a strictly positive lower bound on the 
smallest singular value $t_1(D^\delta_h-z)$ of $(D^\delta_h-z)$ not 
for one specific tunneling potential $Q$, but in fact with 
good probability with respect to the probability measure 
\eqref{eq:probaM}. 
\\
\par 
For $\|\gamma\| < 2R$, we consider the holomorphic function 
\begin{equation}\label{plb:eq3}
	\gamma\mapsto f(\gamma) := \det E_{-+}^\delta(z,Q) 
\end{equation}
Combining \eqref{GP:eqn4.0}, \eqref{GP:eq12} and \eqref{eq:DiracPotential14} 
we see that $\| E_{-+}^\delta\| \leq C\tau_0$. Therefore 
$|\det E_{-+}^\delta|\leq (C\tau )^N$ with $N=\mO(h^{-1})$, and 
\begin{equation}\label{plb:eq4}
	\log | f(\gamma)| \leq  h^{-2} \varepsilon_0(h) , 
	\quad \|\gamma\| < 2R,
\end{equation}
where 
\begin{equation}\label{plb:eq6}
	\varepsilon_0(h) =  C(\log \tau_0^{-1} + (\log h^{-1})^2)(h+h^2\log h^{-1}).
\end{equation}
Furthermore, we deduce from \eqref{plb:eq2} that
\begin{equation}\label{plb:eq5}
	\log | f(\gamma)| \geq - \varepsilon_0(h)h^{-2}.
\end{equation}
for a $\gamma$ satisfying $\|\gamma\|\leq R/2$. 
\\
\par
The relations (\ref{plb:eq4}-\ref{plb:eq5}) correspond to 
\cite[(8.4-8.7)]{Sj09}, and the discussion in Section 8 in \cite{Sj09} 
applies to our case with the obvious modifications,  
and we obtain the analog of \cite[Proposition 8.2]{Sj09}.
\begin{prop}\label{prop:probaLowerBd}
Let $z\in\C$ be fixed, let $s>1$, $\varepsilon\in ]0, s-1[$ and 
let $\kappa_1,\kappa_2,\kappa_3$ be as in Proposition \ref{prop:p2}. 
For $C_L>0$ large enough let $L$ be as in \eqref{prop:p1.2}.
Let $\gamma\in \C^D$ be a real or complex random vector (depending on 
whether the operators $P_j$ in \eqref{eq2} have real or complex 
coefficients) with probability law \eqref{eq:probaM}. %
\par
Put $\kappa_5 = \kappa_3+\kappa_4 +2 + \frac{10}{s-1-\varepsilon}$, 
where $\kappa_4$ is as in \eqref{eq:probaM2}. 
Let $0 < \varepsilon <  \exp( -C_1 h^{-2} \varepsilon_0(h))$, $C_1>1$ 
sufficiently large. Then there exist a $C>1$ such that for $h>0$ 
small enough
\begin{equation*}
	\prob \left( 
		 |\det E_{-+}^\delta(z)  | < \varepsilon
		\right)
	\leq C h^{-\kappa_5}\varepsilon_0(h)
	\exp\left( - \frac{h^2}{C\varepsilon_0(h)} \log \frac{1}{\varepsilon }
	\right).
\end{equation*} 
\end{prop}
Using \eqref{GP:eq15} and the fact that $\| E_{-+}^\delta\| \leq C\tau_0$ we get  
\begin{equation*}
	|\det E_{-+}^{\delta} | 
	\leq 8t_1(D^\delta_h-z) s_1(E_{-+}^{\delta})^{N-1}
	\leq 8(C\tau_0)^{N-1}t_1(D^\delta_h-z) .
\end{equation*}
Hence,
\begin{equation*}
	t_1(D^\delta_h-z) \geq \frac{|\det E_{-+}^{\delta} | }{8(C\tau_0)^{N-1}},
\end{equation*}
and we conclude that under the assumptions of Proposition 
\ref{prop:probaLowerBd}
\begin{equation}\label{eq:FL}
	\prob \left( 
		 t_1(D^\delta_h-z) \geq \frac{\varepsilon}{8(C\tau_0)^{N-1}}
		\right)
	\geq 1- C h^{-\kappa_5}\varepsilon_0(h)
	\exp\left( - \frac{h^2}{C\varepsilon_0(h)} \log \frac{1}{\varepsilon }
	\right),
\end{equation} 
which concludes the proof of Theorem \ref{thm1}. 
\section{Counting eigenvalues}\label{sec:Count}
The dual lattice $\Gamma^*$ of $\Gamma$, see \eqref{eq:Gamma}, is determined by 
the condition that $\gamma^*\in\Gamma^*$ when $\Rea \langle \gamma | \gamma^*\rangle 
\in 2\pi \Z$ for $\gamma \in \Gamma$. This yields that $\Gamma^* = 
\frac{1}{\sqrt{3}}(\omega \Z \oplus \omega^2\Z)$, with $\omega$ as in 
\eqref{eq:Gamma}.
\par
For $h\in]0,1]$, let $D_h^\delta$ be as in \eqref{plb:eq1.1}, and for 
$\rho \in \C$ put  
\begin{equation}\label{eq:ce0}
\begin{split}
	P_\rho &:=  \begin{pmatrix}
		2hD_{\overline{x}} & 0 \\
		0 & 2hD_{\overline{x}} \\
	\end{pmatrix}
	+
	\rho
	\left(
	\begin{pmatrix}
		0& U(x) \\
		U(-x) & 0 \\
	\end{pmatrix}
	+\delta h^{\kappa_1}Q
	\right)
	\\
	&=: \wt{D}_h + \rho \left( \wt{U} + \delta h^{\kappa_1}Q\right),
\end{split}
\end{equation}
which is a closed densely defined unbounded operator 
$L^2(\C/\Gamma;\C^2) \to L^2(\C/\Gamma;\C^2)$ with domain 
$H_h^1(\C/\Gamma;\C^2)$. Note that $P_0 = \wt{D}_h$ and 
$P_1 = D_h^\delta$. Furthermore, $P_\rho$ is Fredholm of index 
$0$ for any $\rho \in \C$, since it is a bounded perturbation of 
a Fredholm operator of index $0$. 
\par 
Since for any $\gamma^*\in \Gamma^*$ 
\begin{equation*}
	\e^{-i\Rea\langle x | \gamma^*\rangle  }P_\rho\,
	\e^{i\Rea\langle x | \gamma^*\rangle  }
	= P_\rho + h\gamma^*,
\end{equation*}
we deduce from the fact that $P_\rho$ is Fredholm of index 
$0$ that 
\begin{equation}\label{eq:ce1}
	\sigma(P_\rho)	= \sigma(P_\rho) + h\gamma^*, \quad 
	\gamma^*\in \Gamma^*.
\end{equation}
By \cite[(2.9)]{BEWZ22} we know that $\sigma(\wt{D}_h) = \sigma(P_0)= h\Gamma^*$.  
So provided that $\lambda\neq h\Gamma^*$, we know that 
$(\wt{D}_h - \lambda)^{-1}:L^2(\C/\Gamma;\C^2) \to L^2(\C/\Gamma;\C^2)$ is a compact 
operator, since $H^1_h(\C/\Gamma;\C^2)$ injects compactly into $L^2(\C/\Gamma;\C^2)$. 
Hence, for $\lambda\neq h\Gamma^*$, the operator 
\begin{equation*}
	T_\lambda := (\wt{D}_h - \lambda)^{-1}\left( \wt{U} + \delta h^{\kappa_1}Q\right)
	:~L^2(\C/\Gamma;\C^2) \to L^2(\C/\Gamma;\C^2) 
\end{equation*}
is compact. So, for $\lambda\neq h\Gamma^*$, we have that 
\begin{equation*}
	\lambda \in \sigma(P_\rho) \quad \Leftrightarrow 
	\quad
	-1/\rho \in \sigma(T_\lambda). 
\end{equation*}
Since $T_\lambda$ is compact, this can happen for at most 
a countable set of points $\rho\in\C$ with a potential accumulation point 
at infinity. 
\par
Until further notice, we assume that the event \eqref{eq:FL} holds. 
Therefore, the spectra of both $P_0$ and $P_1$ are discrete and 
we can find a continuous path $\rho(t)\in \C$, $t\in [0,1]$, 
with $\rho(0)=0$ and $\rho(1)=1$ such that 
\begin{equation*}
	\sigma(P_{\rho(t)}) \text{ is discrete for all } t\in [0,1].
\end{equation*}
For $\gamma^*\in\Gamma^*$ let $\mathfrak{C}_{\gamma^*}$ denote a primitive cell of $\Gamma^*$. Here we 
see $\mathfrak{C}_{\gamma^*}$ as parallelogram in $\C\simeq \R^2$ with edge points $\gamma^*$,$\gamma^*+\omega/\sqrt{3}$,
$\gamma^*+\omega^2/\sqrt{3}$ and $\gamma^*+\omega/\sqrt{3}+\omega^2/\sqrt{3}$, 
for some $\gamma^*$. Moreover, by convention, $\gamma^*$ is the only edge 
point contained in $\mathfrak{C}_{\gamma^*}$, and the boundary of 
$\mathfrak{C}_{\gamma^*}$ is given by the lines 
from $\gamma^*$ to $\gamma^*+\omega/\sqrt{3}$ and from $\gamma^*$ and 
$\gamma^*+\omega^2 /\sqrt{3}$ with the edge points $\gamma^*+\omega/\sqrt{3}$ and 
$\gamma^*+\omega^2/\sqrt{3}$ excluded. A quick computation shows that 
\begin{equation}\label{eq:ce2}
	|\C/\Gamma| = \frac{(4\pi)^2\sqrt{3}}{2}, \quad 
	|\mathfrak{C}_{\gamma^*}| = \frac{1}{2\sqrt{3}}.
\end{equation}
Notice that the primitive cells of $h\Gamma^*$ are of the form $h\mathfrak{C}_{\gamma^*}$. 
The $\Gamma^*$-translation invariance \eqref{eq:ce1} of $\sigma(P_{\rho(t)})$ 
shows that each $h\mathfrak{C}_{\gamma^*}$ contains the exact same 
amount of eigenvalues of $P_{\rho(t)}$ counting multiplicities. Since 
the spectra $\sigma(P_{\rho(t)})$ are discrete, the eigenvalues depend 
continuously on $\rho(t)$ and thus on $t$. Therefore, by continuity, 
for each $\gamma^*$ 
\begin{equation}\label{eq:ce3}
	\#( \sigma(P_{\rho(t)}) \cap h\mathfrak{C}_{\gamma^*}) = 
	\#( \sigma(P_0) \cap h\mathfrak{C}_{\gamma^*}) = 2, 
	\quad \text{for all } t \in [0,1].
\end{equation}
The equality in the last line comes from \cite[(2.9)]{BEWZ22}. 
\par
Let $\Omega\Subset\C$ be relatively compact with a Lipschitz boundary 
and recall from \eqref{eq:ce0} that $P_1 = D_h^\delta$. 
Counting how many fundamental cells $h\mathfrak{C}_{\gamma^*}$ are needed to cover 
$\Omega$ and using \eqref{eq:ce2} we find that 
\begin{equation}\label{eq:ce4}
	\#( \sigma(D_h^\delta) \cap \Omega) = \frac{2|\Omega|}{h^2|\mathfrak{C}_0|} 
	+\mO(h^{-1})
	= \frac{2|\Omega|\cdot |\C/\Gamma| }{(2\pi h)^2} +\mO(h^{-1}).
\end{equation}
On the other hand, the matrix valued symbol $d(x,\xi)$, as in \eqref{eq1.1}, 
has the two complex valued eigenvalues 
\begin{equation}\label{eq:ce5}
	\lambda_{j}(x,\xi) = \xi_1+i\xi_2 + (-1)^j \sqrt{U(x)U(-x)}, \quad (x,\xi)\in T^*(\C/\Gamma), 
	\quad j=1,2.
\end{equation}
So, by Fubini's theorem and translation in $\xi$ we find that
\begin{equation}\label{eq:ce6}
	\iint_{(\lambda_j)^{-1}(\Omega)} dxd\xi 
	= \int_{\C/\Gamma} \int_{\xi_1+i\xi_2 \in \Omega} d\xi dx
	=|\C/\Gamma| \cdot |\Omega|.
\end{equation}
Combining this with \eqref{eq:ce4} we find that 
\begin{equation}\label{eq:ce7}
	\#( \sigma(D_h^\delta) \cap \Omega) 
	= \frac{1}{(2\pi h)^2}\sum_1^2 \iint_{(\lambda_j)^{-1}(\Omega)} dxd\xi  
	+\mO(h^{-1}).
\end{equation}
with probability bounded from below by the right hand side of \eqref{eq:FL}. 
The claim of Theorem \ref{thm3} follows.
\appendix 
\section{Matrix-valued smooth pseudo-differential calculus}
\label{App:MatrixPseudo}
In this section we review some essential facts of semiclassical 
analysis on a compact $d$-dimensional manifold $X$. In what follows 
we equip $\C^n$, $n\in\N^*$, with the standard sesquilinear scalar 
product, and $\mathrm{Hom}(\C^n,\C^n)\simeq 
\C^{n\times n}\ni A$ with the Hilbert-Schmidt scalar product 
$\langle A|B \rangle_{\mathrm{HS}} = (\tr B^*A )^{1/2}$. 
\\
\par
Let $U\subset \R^d$ be an open set. For $n\in\N$ and $m\in\R$ we define the 
symbol class $S^{m}(T^*U;\mathrm{Hom}(\C^n,\C^n))$ as the set of all 
functions $p(\cdot\,; h) \in C^{\infty}(T^*U;\mathrm{Hom}(\C^n,\C^n))$ 
such that for any compact set $K\Subset U$ and for all $\alpha,\beta \in \N^d$ 
there exists a constant $C_{\alpha,\beta,K}>0$ such that for all $h\in]0,1]$
\begin{equation}\label{eq:sc1}
	\|\partial_x^\alpha\partial_\xi^\beta a(x,\xi;h)\|_{\mathrm{HS}} 
	\leq C_{\alpha,\beta,K}\langle \xi \rangle^{m-|\beta|}.
\end{equation}
Here, $\langle \xi \rangle := (1+|\xi|^2)^{1/2}$.  
The choice of norm on $\mathrm{Hom}(\C^n,\C^n)$ is not important 
since all norms are equivalent, so a change of norm would only lead to a 
change of the constant in the symbol estimate. So the symbol class $S^m$ 
defined above is independent of the choice of norm on $\mathrm{Hom}(\C^n,\C^n)$. 
One can easily check that the above definition of the symbol class is 
equivalent to having all entries $a_{ij}$ of 
the matrix valued symbol in the standard symbol class 
$S^{m}(T^*U)$. The residual symbol space of order $-\infty$ is defined by 
$S^{-\infty}(T^*U;\mathrm{Hom}(\C^n,\C^n)):= 
\bigcap_{m}S^{m}(T^*U;\mathrm{Hom}(\C^n,\C^n))$. 
\\
\par 
Given the above definition of the class $S^{m}(T^*U;\mathrm{Hom}(\C^n,\C^n))$, 
we can also define symbols on $T^*X$. We say that $p\in S^m(T^*X;\mathrm{Hom}(\C^n,\C^n))$ if 
$p\in C^\infty(T^*X;\mathrm{Hom}(\C^n,\C^n))$ and if for any smooth coordinate chart 
$\phi:X\supset U \to V\subset \R^d$ we have that 
$(\widehat{\phi}^{-1})^*p \in S^m(T^*V;\mathrm{Hom}(\C^n,\C^n))$, 
where $\widehat{\phi}^{-1}:T^*V \to T^*U$ is defined 
by $\widehat{\phi}^{-1}(x,\xi)=(\phi^{-1}(x),(d_x\phi^{-1})^{-T}\xi)$. 
\\
\par 
Let $L^{\infty}(X;\C^{n\times n})$ be the space $\C^{n\times n}$-valued 
maps $f$ on $X$ such that $\| \sqrt{\tr f^* f} \|_{L^{\infty}(X)}$. 
A linear continuous map $R_h: \mathcal{D}'(X;\C^n) \to C^{\infty}(X;\C^n)$ 
is called \emph{negligible} if its distribution kernel $K_R$ is smooth and 
each of its $C^\infty(X\times X;\mathrm{Hom}(\C^n,\C^n))$ seminorms is 
$\mO(h^\infty)$, i.e. it satisfies 
\begin{equation}\label{eq:sc2}
\|\partial_x^\alpha \partial^\beta_y K_R(x,y)\|_{L^\infty(X;\C^{n\times n})}
	= \mO(h^\infty), 
\end{equation}
for all $\alpha,\beta\in \N^d$, when expressed in local coordinates.
\\
\par
Let $\phi: X\supset U \to V\subset\R^d$ be a coordinate chart between open 
sets and let $\chi\in C_c^\infty(U)$. We will refer to 
the induced pair $(\phi,\chi)$ as a \emph{cut-off chart}. 
\begin{definition}
A linear continuous map $P_h: C_c^{\infty}(X;\C^n)\to \mathcal{D}'(X;\C^n)$ 
is called a \emph{semiclassical pseudo-differential operator} belonging 
to the space $\Psi_{h}^{m}(X;\C^n,\C^n)$ if and only if
\begin{itemize}
	\item $\chi P_h \psi$ is negligible for all 
	$\chi,\psi \in C^\infty_c(X)$ with $\supp \chi \cap \supp \psi = \emptyset$;
	\item for every coordinate chart $\phi:  U \to V$ there exists a symbol 
	$p_{\phi}\in S^{m}(T^*\R^d;\mathrm{Hom}(\C^n,\C^n))$ 
	such that for all $\chi,\psi \in C^\infty_c(U)$ 
	\begin{equation}\label{eq:sc3}
		\chi P_h\psi 
		= \chi \phi^*\Op_h(p_{\phi})(\phi^{-1})^*\psi.
	\end{equation}
\end{itemize}
\end{definition}
In \eqref{eq:sc3} we use the standard semiclassical quantization of 
symbols $a\in S^{m}(T^*\R^d;\mathrm{Hom}(\C^n,\C^n))$ defined by 
\begin{equation}\label{eq:sa4}
	\Op_h(a) u(x)= a(x,hD_x)u(x) = \frac{1}{(2h\pi)^d} \iint_{\R^{2d}} 
	\e^{\frac{i}{h} (x-y)\cdot \xi} 
	a(x,\xi;h)u(y) dy d\xi, 
	\quad u\in C_c^\infty(V),
\end{equation}
seen as an oscillatory integral. 
\\
\par 
Given a symbol $p\in S^{m}(T^*X;\mathrm{Hom}(\C^n,\C^n))$ one can obtain 
an operator $\Op_h(p) \in \Psi_{h}^{m}(X;\C^n,\C^n)$, using a non-canonical 
quantization procedure: Take a partition of 
unity $\{\chi_k\}_{k \in K}$ subordinate to a finite covering of $X$ 
by coordinate charts $\{\phi_k:X\supset U_k\to V_k\subset \R^d\}_{k\in K}$ 
such that $\sum \psi^2_k =1$. Then
\begin{equation}\label{eq:sc4}
	\Op_h(p) = \sum_{k \in K} \chi_k\phi_k^* 
	\Op_h(p_{\phi_k}) 
	(\phi_k^{-1})^* \chi_k ~ \in \Psi_{h}^{m},
	\quad 
	p_{\phi_k} := p\circ\widehat{\phi}_k^{-1}
\end{equation}
where $\widehat{\phi}_k^{-1}:T^*V_k \to T^*U_k$ is defined 
as above. The correspondence $\Psi_{h}^{m}\ni P_h \mapsto p\in S^m$ 
is not globally well-defined, but it gives rise to a bijection 
\begin{equation}\label{eq:sc5}
	\Psi_{h}^{m}(X;\C^n,\C^n)/ h\Psi_{h}^{m-1} (X;\C^n,\C^n)
	\longrightarrow 
	S^{m}(T^*X;\mathrm{Hom}(\C^n,\C^n)) / hS^{m-1}(T^*X;\mathrm{Hom}(\C^n,\C^n)) .
\end{equation}
The image $\sigma_P$ of $P$ under the map \eqref{eq:sc5} is called 
\emph{principal symbol} of $P$.
\\
\par
A pseudo-differential operator $P \in \Psi_{h}^{m}(X;\C^n,\C^n)$ maps 
$C^\infty(X;\mathrm{Hom}(\C^2,\C^2))$ to $C^\infty(X;\mathrm{Hom}(\C^2,\C^2))$ 
and so extends to a well-defined operator $\mathcal{D}'(X;\C^n)\to 
\mathcal{D}'(X;\C^n)$. They can therefore be composed with each other, and one 
can easily check that if $P_j\in \Psi_{h}^{m_j}(X;\C^n,\C^n)$, $j=1,2$, then 
$P_1\circ P_2 \in \Psi_{h}^{m_1+m_2}(X;\C^n,\C^n)$ and 
\begin{equation*}
	\sigma_{P_1\circ P_2}(x,\xi) = \sigma_{P_1}(x,\xi)\sigma_{P_2}(x,\xi).
\end{equation*}
\par
We recall the definition of \emph{semiclassical Sobolev spaces} on $X$. 
Let $\cV$ be either $\cV=\C^n$ or $\cV=\C^{n\times n}$ equipped with a 
sesquilinear scalar product $\langle \cdot | \cdot \rangle_{\cV}$ and norm 
$\|\cdot\|_{\cV}$ as discussed at the beginning of Appendix \ref{App:MatrixPseudo}. 
First we recall that on $\R^d$ we define the \emph{semiclassical Sobolev space} 
$H_{h}^s(\R^d;\cV)\subset \mathcal{S}'(\R^d;\cV)$, $s\in\R$, as the space of 
all tempered distributions $u\in\mathcal{S}'(\R^d;\C^n)$ such that 
\begin{equation*}
	\| u \|_{H_{h}^s(\R^d;\cV)} 
	= \| \Lambda^s u \|_{L^2(\R^d;\cV)} 
	= \| \|\Lambda^s u\|_{\cV} \|_{L^2(\R^d)} 
	<\infty.
\end{equation*}
where $\Lambda^s:=\diag( \langle hD_x \rangle^s,\dots, \langle hD_x \rangle^s):
\C^n \to \C^n$. The inclusion $H^s_h\subset  H^{s'}_h$, $s\geq s'$, is 
compact when $s>s'$, since this is true component wise. 
Via the standard $L^2$ pairing and the identification $(\C^n)^*\simeq \C^n$ 
via the standard sesquilinear scalar product on $\C^n$, we 
can identify $H_{h}^{-s}$ with the dual of $H^{s}_h$. 
\\
\par 
Let $\{\phi_k:X\supset U_k\to V_k\subset \R^d\}_{k\in K}$ 
be a finite collection of coordinate charts covering $X$, 
and $\{\chi_k\}_{k\in K}$, $\chi_k \in C^\infty_c(U_k)$, 
be a finite partition of unity subordinate to the covering 
by coordinate charts. We define the \emph{semiclassical Sobolev 
space} $H_{h}^s(X;\cV)\subset \mathcal{D}'(X;\cV)$ 
as the set of all distributions $u\in\mathcal{D}'(X;\cV)$ such that 
\begin{equation*}
		(\phi^{-1})^*\chi u \in H_{h}^s(\R^d;\cV)
\end{equation*}
for all cut-off charts $(\phi,\chi)$. We can equip $H_{h}^s(X;\cV)$ 
with the norm
\begin{equation}\label{eq:sc7}
	\| u \|_{H_{h}^s(X;\cV)}^2 = \sum_{k\in K}
	\|(\phi_k^{-1})^*\chi_k u \|_{H_{h}^s(\R^d;\cV)}^2, \quad 
	K< +\infty, 
\end{equation}
making it Banach space. Taking different coordinate patches and 
cut-off functions in \eqref{eq:sc7} yields an equivalent (uniformly 
as $h\to 0$) norm. Moreover, 
these spaces satisfy the inclusion $H_h^s\subset H_h^0 \subset H_h^{-s}$, 
$s>0$, the compact injection $H^s_h \hookrightarrow H^{s'}_h$, $s>s'$. 
See for instance \cite[Chapter 1, \S 7]{Sh01} (there the discussion is 
for $h=1$ but the case of general $h$ is completely analogous). 
\par
One may follow the standard proofs from the case when $n=1$ to obtain the 
following continuity result.
\begin{prop}\label{app:prop6}
	Let $P\in \Psi_{h}^{m}(X;\C^n,\C^n)$. Then, the operator 
	\begin{equation}\label{eq:sa8}
		P:~H_{h}^{s}(X;\C^n) \longrightarrow H_{h}^{s-m}(X;\C^n),
	\end{equation}
	is bounded uniformly in $h>0$. 
\end{prop}
If we equip $X$ with a smooth density of integration, denoted by 
$dx$, then we can define an $L^2$ scalar product by 
\begin{equation}\label{eq:sc7.0}
	\langle \phi | \psi \rangle 
	= 
	\int_X \langle \phi(x) | \psi(x) \rangle_{\cV} \,dx, 
	\quad 
	\phi,\psi \in C^\infty(X;\cV),
\end{equation}
The induced norm is equivalent to the norm of $H_{h}^0(X;\cV)$ defined 
in \eqref{eq:sc7} and we get $L^2(X;\cV)\simeq H^0_h(X;\cV)$. Here, the 
$L^2$ space is defined as the closure of $C^\infty(X;\cV)$ under the 
norm induced by \eqref{eq:sc7.0}. Taking a different density of integration 
or inner product on $\cV$ in \eqref{eq:sc7.0} yields again 
an equivalent norm. 
\par
An equivalent way to define the Sobolev spaces $H_{h}^s(X;\cV)$ is as follows: 
Consider a family of non-negative elliptic differential operators
\begin{equation*}
	h^2R_\ell = \sum_{j,k=1}^d (hD_{x_j})^*r_{j,k,\ell}(x) hD_{x_k}, 
	\quad \ell =1,\dots,n,
\end{equation*}
expressed in local coordinates, where the star indicates that we take the 
adjoint with respect to some fixed positive smooth density on $X$. Each 
$h^2R_\ell$ is selfadjoint with domain $H^2_h(X;\C)$, so $(1+h^2R_\ell)^{s/2}
:L^2(X;\C)\to L^2(X;\C)$ is a closed densely defined linear operator 
for $s\in\R$. It is bounded precisely when $s\leq 0$.
\begin{prop}[{\cite[Proposition 16.2.2]{Sj19}}]\label{app:prop9}
		For every $s\in\R$ and $\ell\in \{1,\dots,n\}$ we have that $H^s_h(X;\C)$ 
		is the space of all 
		$u\in \mathcal{D}'(X;\C)$ such that $(1+h^2R_\ell)^{s/2} u \in L^2(X;\C)$ 
		and the norm $\|u \|_{H_{h}^s(X;\C)}$ is equivalent to the norm 
		$\|(1+h^2R_\ell)^{s/2}u \|_{L^2(X;\C)}$, uniformly as $h\to 0$. 
\end{prop}	
We deduce that for every $s\in\R$ the space $H^s_h(X;\cV)$ is the space of all 
$u\in \mathcal{D}'(X;\cV)$ such that $\Lambda^s u \in L^2(X;\cV)$ where 
\begin{equation*}
	\Lambda^s:=\diag((1+h^2R_1)^{s/2}, \dots, (1+h^2R_n)^{s/2})
\end{equation*}
Furthermore, the norm $\|u \|_{H_{h}^s(X;\cV)}$ is equivalent to the norm 
$\|\Lambda^s u \|_{L^2(X;\cV)}$, uniformly as $h\to 0$. Using the scalar 
product \eqref{eq:sc7.0}, we also recover the fact 
that the dual of $H^s_h$ can be identified with $H^{-s}_h$.
\\
\par
Under the above assumptions the spectrum of $R_\ell$, $\ell =1,\dots, n$,  
consists solely out of isolated eigenvalues with finite multiplicities. We 
denote the respective eigenvalues by 
\begin{equation*}
		0\leq \lambda_1(\ell)^2 \leq \lambda_2(\ell)^2 \leq \dots, 
		\quad 0\leq \lambda_j(\ell),
\end{equation*}
ordered increasingly and repeated according to multiplicities. For 
each $\ell =1,\dots,n$, let $\phi_j^\ell$, $j\in \N^*$, denote a corresponding  
$L^2$ orthonormal basis of eigenfunctions. Thanks to the Weyl asymptotics for the eigenvalues 
of elliptic second order selfadjoint differential operators on compact manifolds, see for 
instance \cite[Chapter 12]{GriSjoMLA} and the reference given therein, we know 
that for each fixed $\ell =1,\dots,n$
\begin{equation}\label{eq:sa8.00}
	\#\{ j\in\N^*; \lambda_j(\ell)\leq \lambda \} 
	\asymp \lambda^2, 
\end{equation}
uniformly for $\lambda\geq 1$ and independently of $h$.
In fact much more is known about these asymptotics, see for 
instance \cite{Iv98}, however we do not need such refinements here. 
%
By elliptic regularity we know that $\phi_j^\ell \in C^\infty(X)$, and 
so if $u\in \mathcal{D}'(X)$ then 
\begin{equation*}
	u = \sum_k u(k,\ell) \phi_k^\ell, \quad u(k,\ell) = (u,\phi_k^\ell), 
\end{equation*}
in $\mathcal{D}'(X)$. Since $X$ is compact, the sequence of coefficients 
is of at most temperate growth. Fix $\ell\in \{1,\dots,n\}$ and write 
$\mu_n(\ell)=h\lambda_n(\ell)$. Proposition \ref{app:prop9} and the above 
discussion then imply that $u\in H^{s}_h(X;\C^n)$ if and only if 
$\sum_{\ell=1}^n\sum_k\langle \mu_k(\ell)\rangle^{s}|u_\ell(k,\ell)|^2$ 
is finite and 
\begin{equation*}\label{eq:sa8.0}
	\|u \|_{H_{h}^s(X;\C^n)}^2 \asymp 
	\sum_{\ell=1}^n\sum_k\langle \mu_k(\ell)\rangle^{2s}|u_\ell(k,\ell)|^2.
\end{equation*}
Similarly, using the Hilbert-Schmidt scalar product on $\C^{n\times n}$ we 
see that $u\in H^{s}_h(X;\C^{n\times n})$ if and only if 
$\sum_{i,j=1}^{n}\sum_k\langle \mu_k(i)\rangle^{s}|u_\ell(k,i,j)|^2$ 
is finite and 
\begin{equation}\label{eq:sa8.1}
	\|u \|_{H_{h}^s(X;\C^{n\times n})}^2 \asymp 
	\sum_{i,j=1}^{n}\sum_k\langle \mu_k(i)\rangle^{2s}|u_\ell(k,i,j)|^2.
\end{equation}
uniformly in $h>0$.
%
%
\begin{prop}\label{app:prop5} 
	Let $X$ be a smooth compact manifold equipped with a smooth 
	density of integration and let $P \in \Psi_{h}^{m}(X;\C^n,\C^n)$ 
	then its formal $L^2$ adjoint $P^* \in \Psi_{h}^{m}(X;\C^n,\C^n)$ 
	with 
	\begin{equation}\label{eq:sc7.1}
		\sigma(P^*) = \sigma(P)^*,
	\end{equation}
	where the adjoint on the right hand side is the adjoint of $n\times n$ 
	complex matrices with respect to the standard sesquilinear inner 
	product $\langle \cdot | \cdot \rangle_{\C^n}$. 
\end{prop}
The following Sobolev estimates can be obtained through straight forward 
computations when $X=\R^d$, and can be carried over to the case of a compact 
smooth manifold $X$ via \eqref{eq:sc7}. See for instance 
\cite[Section 16.1]{Sj19}. 
\begin{prop}\label{app:prop4} 
	Let $s>d/2$. Then there exists a constant $C=C(s)>0$ such that for 
	all $u,v \in H_{h}^s(X;\C)$, $w\in H^s(X;\C)$ we have that $u\in L^\infty(X)$, 
	$uv \in H_{h}^s(X;\C)$, $uw \in H_{h}^s(X;\C)$ and 
	\begin{equation*}
		\begin{split}
		&\|u \|_{L^\infty(X)} \leq C h^{-d/2}\| u \|_{H_{h}^s(X;\C)} \\
		&\|uv \|_{H_{h}^s(X;\C)} 
			\leq C h^{-d/2}\| u \|_{H_{h}^s(X;\C)} 
					\| v\|_{H_{h}^s(X;\C)},\\
		&\|uw \|_{H_{h}^s(X;\C)} 
			\leq C\| u \|_{H_{h}^s(X;\C)} 
					\| w\|_{H^s(X;\C)}.
		\end{split}
	\end{equation*}
\end{prop}
\begin{definition}\label{def:ClassEllip}
	A symbol $p \in S^{m}(T^*X;\mathrm{Hom}(\C^n,\C^n))$ is called 
	classically elliptic if there exists a symbol 
	$q\in S^{-m}(T^*X;\mathrm{Hom}(\C^n,\C^n))$ such that 
	$pq-1$, $qp-1\in S^{-1}(T^*X;\mathrm{Hom}(\C^n,\C^n))$. 
\end{definition}
\begin{prop}
	Let $P\in \Psi_{h}^{m}(X;\C^n,\C^n)$ with classically elliptic principal 
	symbol. Then there exists a $C>0$ such that for all $u\in H^s_h$ 
	\begin{equation}\label{eq:sa9}
		\|u\|_{H^{s}_h} \leq C(\|Pu\|_{H^{s-m}_h} + \|u\|_{H^0_h}).
	\end{equation}
\end{prop}
Following \cite[Chapter 9]{DiSj99}, with the obvious modifications 
for the case of matrix-valued symbols, we get the following properties 
for trace class operators. 
\begin{prop}\label{app:prop3}
Let $a\in S^{m}(T^*\R^d;\mathrm{Hom}(\C^n,\C^n))$ such that 
\begin{equation}\label{eq:sa10}
	\sum_{|\alpha|\leq 2d+1} \| \partial^\alpha_\rho a_{jk} \|_{L^1} < \infty, 
	\quad j,k =1,\dots, n. 
\end{equation}
Then the corresponding operator $a(x,hD_x)$ is of trace class and 
\begin{equation}\label{eq:sa11}
	\| a(x,hD_x)\|_{\tr} = O(h^{-d}) 
	\sum_{j,k=1}^n\sum_{|\alpha|\leq 2d+1}
	\iint | (h^{1/2}\partial_{x,\xi})^\alpha a_{jk}(x,\xi)| dx d\xi, 
\end{equation}
\begin{equation}\label{eq:sa12}
	\tr a(x,hD_x) = \frac{1}{(2\pi h)^d} \iint \tr a(x,\xi) dx d\xi, 
\end{equation}
where the second trace is the trace of complex $n\times n$ matrices. 
\end{prop}

\providecommand{\bysame}{\leavevmode\hbox to3em{\hrulefill}\thinspace}
\providecommand{\MR}{\relax\ifhmode\unskip\space\fi MR }
\providecommand{\MRhref}[2]{%
  \href{http://www.ams.org/mathscinet-getitem?mr=#1}{#2}
}
\providecommand{\href}[2]{#2}

\end{document}